\documentclass[a4paper,UKenglish,cleveref,autoref,thm-restate,numberwithinsect]{lipics-v2021}

\usepackage{mathtools}
\usepackage{nicefrac}
\usepackage{hyphenat}
\usepackage{booktabs}
 
\usepackage{tikz}
\usetikzlibrary{arrows.meta,graphs,shadows,patterns,positioning,calc,backgrounds,fit}
\usepackage{pgfplots}
\pgfplotsset{compat=1.18}
\usepackage{pgfplotstable}
 
\usepackage{xspace}
 
\newtheorem{rr}[theorem]{Reduction Rule}
 
\Crefname{rr}{Rule}{Rules}
\crefname{rr}{Rule}{Rules}
\Crefname{theorem}{Theorem}{Theorems}
\crefname{theorem}{Thm.}{Thms.}
\Crefname{lemma}{Lemma}{Lemmata}
\crefname{lemma}{Lem.}{Lems.}
\Crefname{observation}{Observation}{Observations}
\crefname{observation}{Obs.}{Obs.}
 
\newcommand{\appendixstar}{\ensuremath{\star}\xspace}
 
\newcommand{\lemDP}{}
\newcommand{\lemMAPPDtable}{}
\newcommand{\lemMxTPDtable}{}
\newcommand{\lemMnTPDtable}{}
 
\DeclareRobustCommand{\tikzdot}[1]{\tikz[baseline=-0.6ex]{\node[draw,fill=#1,inner sep=2pt,circle] at (0,0) {};}}

\newcommand{\Oh}{\ensuremath{\mathcal{O}}}

\DeclareMathOperator{\twwithoutN}{{tw}}
\DeclareMathOperator{\swwithoutN}{{sw}}
\DeclareMathOperator{\nswwithoutN}{{nsw}}

\DeclareMathOperator{\anc}{anc}
\DeclareMathOperator{\DP}{DP}

\DeclareMathOperator{\GW}{GW}
\DeclareMathOperator{\parents}{parents}
\DeclareMathOperator{\children}{children}

\newcommand{\w}{{\ensuremath{\omega}}\xspace}
\newcommand\lb{\linebreak}
\newcommand\Recc[1]{Recurrence~(\ref{#1})}

\newcommand{\Wh}[1]{{\normalfont W[#1]}\xspace}
\newcommand{\NP}{{\normalfont{NP}}\xspace}

\newcommand{\Instance}{{\ensuremath{\mathcal{I}}}\xspace}
\newcommand{\Net}{{\ensuremath{{N}}}\xspace}
\newcommand{\Tree}{{\ensuremath{{T}}}\xspace}

\newcommand{\tw}{{\ensuremath{\twwithoutN_{\Net}}}\xspace}

\newcommand{\nsw}{{\ensuremath{\nswwithoutN_{\Net}}}\xspace}
\newcommand{\PD}{\PDsub\Tree}
\newcommand{\PDsub}[1]{{\mathrm{PD}_{#1}}\xspace}
 
\newcommand{\indeg}{{\ensuremath{\deg^-}}\xspace}
\newcommand{\outdeg}{{\ensuremath{\deg^+}}\xspace}
 
\newlength{\boxindent}
\newcommand{\problemdef}[3]{%
	\begin{trivlist}
		\item \tikz{\node [draw,inner sep=4.5pt] {\begin{minipage}{\textwidth-9.4pt}
					\raggedright
					\normalsize\textsc{#1}
 
					\smallskip
 
					\settowidth{\boxindent}{\textbf{Compute:} }%
					\leavevmode\hangindent=\boxindent
					\rlap{\textbf{Input:}}\phantom{\textbf{Compute:}}~#2
 
					\leavevmode\hangindent=\boxindent
					\textbf{Compute:}~#3
			\end{minipage}};}
	\end{trivlist}
}
 
\newcommand{\PROB}[1]{{{\normalfont\textsc{#1}}}\xspace}

\newcommand{\KP}{\PROB{Knapsack}}

\newcommand{\ILP}{\PROB{Integer Linear Program}}

\newcommand{\scanwidth}{\textsf{scanwidth}\xspace}
\newcommand{\panda}{\textsf{PaNDA}\xspace}
\newcommand{\phylozoo}{\textsf{PhyloZoo}\xspace}
 
\newcommand{\NSWTE}{\PROB{Node Scanwidth tree-extension}}

\newcommand{\MAPPD}{\PROB{MAP-PD}}
\newcommand{\MnTPD}{\PROB{Min-Tree-PD}}
\newcommand{\MxTPD}{\PROB{Max-Tree-PD}}

\newcommand{\MAPPDlong}{\PROB{Budgeted-Maximize-All-Paths-PD}}
\newcommand{\MnTPDlong}{\PROB{Minimum-Switching-Tree-PD}}
\newcommand{\MxTPDlong}{\PROB{Budgeted-Maximum-Switching-Tree-PD}}

\newcommand{\bMAPPD}{\PROB{b-\MAPPD}}

\newcommand{\bMxTPD}{\PROB{b-\MxTPD}}

\newcommand{\SbMAPPD}{{\normalfont\texttt{Optimize b-\MAPPD}}\xspace}
\newcommand{\SbMxTPD}{{\normalfont\texttt{Optimize b-\MxTPD}}\xspace}
\newcommand{\SbMnTPD}{{\normalfont\texttt{Compute \MnPD($A$)}}\xspace}
 
\newcommand{\MAP}{\ensuremath{\text{PD}^{\text{map}}_N}}
\newcommand{\MxPD}{\ensuremath{\text{PD}^{\text{max}}_N}}
\newcommand{\MnPD}{\ensuremath{\text{PD}^{\text{min}}_N}}

\newcommand{\B}{\ensuremath{\overline{B}}\xspace}

\title{Tractable Maximization of Budgeted Phylogenetic Diversity on Networks Utilizing Node Scanwidth}
\titlerunning{Budgeted Phylogenetic Diversity on Networks}

\author{Niels Holtgrefe}{Delft Institute of Applied Mathematics, Delft University of Technology, Delft, The Netherlands \and \url{www.NielsHoltgrefe.nl}}{n.a.l.holtgrefe@tudelft.nl}{https://orcid.org/0009-0001-6162-9668}{}
\author{Jannik Schestag}{Department of Informatics, University of Bergen, Bergen, Norway \and \url{www.JannikSchestag.eu}}{j.schestag@uib.no}{https://orcid.org/0000-0001-7767-2970}{}

\authorrunning{N. Holtgrefe and J. Schestag}
\Copyright{Niels Holtgrefe and Jannik Schestag}



\hideLIPIcs
\nolinenumbers

\begin{document}

\maketitle

\begin{abstract}
	\textbf{Background:} Identifying a subset of taxa that maximizes phylogenetic diversity is a cornerstone of quantitative conservation planning.
	Traditionally, phylogenetic diversity is defined over a phylogenetic tree in which leaves resemble present-day taxa and the branch lengths capture the estimated evolutionary distinctiveness.
	While maximizing phylogenetic diversity is computationally tractable on trees with unit costs, the problem becomes computationally intractable when transitioning to phylogenetic networks or to budgeted versions in which protecting taxa incurs non-homogeneous costs.
	This paper addresses these two challenges together, providing definitions and a comprehensive analysis of three distinct variants of budgeted phylogenetic diversity on networks.
	
	\textbf{Results:} We conduct our study through the lens of a small structural parameter, node scanwidth ($\nswwithoutN$), which measures the tree-likeness of a phylogenetic network.
	Given a tree-extension of width $\nswwithoutN$, we show that two of the considered variants can be optimized in $\Oh^*(3^{\nswwithoutN} \cdot B^2)$ time, where $B$ is the budget.
	For the computationally harder third variant, we provide an algorithm to compute phylogenetic diversity scores in $\Oh^*(4^{\nswwithoutN})$ time.
	We further contribute the first exact algorithms to compute node scanwidth itself.
	On highly reticulated, simulated networks with several hundred taxa and heterogeneous costs, our implementation computes phylogenetic diversity scores and optimal node scanwidth in fractions of a second, and the budgeted optimization algorithms significantly outperform existing benchmarks previously limited to unit-cost scenarios.
	
	\textbf{Conclusions:} Node scanwidth proves to be a small and computable parameter that makes budgeted phylogenetic diversity tractable on realistic networks, scaling comfortably to a thousand taxa.
	This narrows the gap between realistic models of reticulate evolution and the computational tools available for conservation planning.
\end{abstract}
\newpage

\section{Background}
\label{sec:introduction}

The biosphere is currently undergoing what scientists describe as the Sixth Mass Extinction~\mbox{\cite{barnosky2011has,cowie2022sixth}}.
Unlike previous events driven by natural phenomena, the current decline in biodiversity is largely anthropogenic, resulting from habitat destruction, climate change, and over-exploitation.
In this urgent context, conservation biologists face the daunting ``agony of choice'': Given limited resources, which species should be prioritized for protection?~\cite{vane}

Functional diversity---the range of biological traits and roles within an ecosystem---is often considered the gold standard for prioritization~\cite{margules2000systematic,sarkar2006biodiversity}.
However, collecting exhaustive trait data for hundreds of thousands of species remains a major bottleneck~\cite{FAITH1992}.
Phylogenetic diversity~(PD) has thus emerged as a desired proxy, operating on the principle that the evolutionary history captured in a phylogeny reflects the underlying feature diversity of a clade.
Classically, the basis of PD is a phylogenetic tree in which the leaves represent present-day species and edge weights can model phylogenetic differences between species.
The PD score of a set of species~$A$ then is the total weight of the edges between the root of the tree and any species in~$A$~\cite{FAITH1992}.

Fortunately, PD can be maximized with a greedy algorithm on trees~\cite{FAITH1992,Pardi2005,steel}, such that even very large instances can be solved optimally within seconds~\cite{minh}.
Several extensions that make the problem more realistic were introduced, but always at the cost of computational complexity.
One of the first generalizations introduced was adding integer costs to species, which however also generalizes the weakly \NP-hard \KP problem~\cite{Pardi2007}.

While PD on trees is well-understood, biological reality is rarely strictly bifurcating~\cite{huson2010phylogenetic}. Evolutionary processes such as hybridization, horizontal gene transfer, and incomplete lineage sorting necessitate the use of phylogenetic networks~\cite{huson2010phylogenetic,LinzThesis}; see~\Cref{fig:example-PD-1}. By moving beyond trees, PD measures can account for the reticulate nature of evolution, capturing the unique contributions of taxa with multiple ancestries. This ensures that conservation priorities reflect the true genetic heritage and evolutionary potential stored within a complex ecosystem, rather than an oversimplified approximation of a tree.
On networks, however, even simple variants of PD maximization become \NP-hard~\cite{bordewichNetworks,AVGTree,WickeFischer2018}.
Moreover, while budgetary constraints (in its most general form the Generalized Noah's Ark Problem) have been studied on trees~\cite{GNAP,Pardi2007}, their integration with network models remains unexplored.

\paragraph*{Our contribution}
In this paper, we address this gap by considering budgeted PD on networks, where taxa are assigned integer costs.
As optimization of PD on trees with a budget is already \NP-hard~\cite{Pardi2007}, and many variants of PD on networks are \Wh{1}-hard with respect to the solution size~\cite{MAPPD,MaxNPD}, we cannot hope for tractable algorithms with respect to tree-likeness or the budget alone.
Instead, we propose a combined parameterization with respect to the budget plus the tree-likeness of the network, captured by its \emph{node scanwidth}.
Informally, node scanwidth measures the maximum number of ancestral lineages that are simultaneously `active' across any horizontal cut of the network as it is processed from the leaves upward~\cite{BruchholdThesis,berry2020scanning}, so that tree-like networks have small node scanwidth and densely reticulate ones have large node scanwidth.
It has so far hardly been considered in the literature~\cite{BruchholdThesis,twvssw}, and computing it is a prerequisite for everything that follows. We therefore begin, in~\Cref{sec:compute-nsw}, by giving reduction rules and an exact algorithm that compute an optimal node scanwidth together with its decomposition.
 
With this decomposition in hand, in~\Cref{sec:apply}, we turn to the diversity problems themselves.
We introduce three budgeted variants of PD maximization on networks, in which taxa may carry different costs of protection, and show how to solve them efficiently in the combined parameter of budget~($B$) and node scanwidth~($\nsw$).
For two of the variants we give $\Oh^*(3^{\nsw} \cdot \min^2(B, \B))$ time algorithms for the optimization, where $\B$ is the total cost of all taxa minus the budget. Parameterizing by $\min(B, \B)$ rather than $B$ alone means the running time is governed by whichever of $B$ and $\B$ is smaller, which is the quantity that is small in practice.
These results also imply that the corresponding unbudgeted problems are `fixed-parameter tractable' in node scanwidth alone, when given an optimal node scanwidth extension.
For the computationally harder third variant, we show that diversity values can be computed in~$\Oh^*(4^\nsw)$ time.
 
Finally, in~\Cref{sec:software} we provide a full implementation of all four algorithms and validate them experimentally.
Our tests confirm that node scanwidth can be computed efficiently on large and realistic networks, and that the resulting diversity implementation improves upon the existing algorithm in~\cite{PaNDA}, even on non-binary networks.

\paragraph*{Related work} 
In recent years, phylogenetic diversity on networks has been studied extensively, with various measures and complexity results established for rooted and semi-directed models~\cite{bordewichNetworks,MAPPD,MAPPDDviability,AVGTree,MaxNPD,AVGTree2,WickeFischer2018}.
While theoretical progress is significant, practical implementations remain scarce.
To our knowledge, the only scalable software packages are \panda~\cite{PaNDA} and an implementation for a PD variant in~\cite{AVGTree} for which we consider lower and upper bounds.
Leveraging bounded scanwidth, \panda can optimize unit-cost instances with hundreds of taxa in a matter of minutes~\cite{PaNDA}, while in~\cite{AVGTree} a bigger parameter, level, is considered.

\section{Preliminaries}
\label{sec:prelims}

Some statements with technical proofs are marked with \appendixstar. Their full proofs are deferred to the appendix; in most cases, a proof sketch is provided in the main text.

For any integer~$n$, we let~$[n] := \{1,\dots,n\}$ and~$[n]_0 := [n] \cup \{0\}$.
We use~$\uplus$ to denote disjoint unions.
We extend functions~$f: A \to B$ and~$g: A \to \mathbb{N}$ to subsets~$A' \subseteq A$ by~$f(A') := \bigcup_{a\in A'} f(a)$ and by~$g_\Sigma(A') := \sum_{a\in A'} g(a)$, where $A$ is an arbitrary set and $B$ a family of sets.
We use the standard~$\Oh$-notation and~$\Oh^*$, in which in the latter, polynomial factors are ignored.

We use standard graph notation~\cite{diestel2025graph}. For a directed graph~$G = (V, E)$, we write~$\indeg(v)$ and~$\outdeg(v)$ for the in-degree and out-degree of~$v \in V$, respectively. 
We write~$E_v^- := \{ uv \in E \mid u \in V \}$ for the set of incoming edges at~$v$ and $V(F) := \{ u, v \mid uv \in F\}$ for the endpoints of an edge set~$F \subseteq E$.
For some~$V' \subseteq V$, we let~$G[V']$ be the subgraph induced by~$V'$.
For a vertex~$v$, we write~$G-v$ for~$G[V\setminus \{v\}]$.

\paragraph*{Phylogenetic networks}
\looseness=-1
A \emph{(rooted, edge-weighted) phylogenetic network}~$\Net = (V, E, \w)$ on a set of taxa~$X$ is a weakly connected, directed acyclic graph (DAG) such that (i) there is a unique vertex with in-degree zero (the \emph{root}) and its out-degree is at least two; (ii) the vertices with out-degree zero (the \emph{leaves}) have in-degree one and they are bijectively labeled by elements of $X$; (iii) all other vertices have degree at least three and in-degree one (\emph{tree-vertices}), or out-degree one (\emph{reticulation vertices}); (iv) $\w: E \to \mathbb{N}_{\geq 0}$ assigns a non-negative \emph{weight} to each edge; see \Cref{fig:example-PD-1}. Since taxa and leaves are in bijection, we refer to them interchangeably.
We write~$X(G)$ for the set of leaves of a directed graph~$G$---not only for networks.
A \emph{phylogenetic tree} is a phylogenetic network without reticulation vertices.
A \emph{chain vertex} is a vertex with in-degree and out-degree~1.
Chain vertices do not occur in phylogenetic networks; however, they may occur after some algorithmic steps defined in this work.

For an edge~$uv$ in a graph~$G$, we say~$u$ is a $G$-\emph{parent} of~$v$ and~$v$ is a $G$-\emph{child} of~$u$.
We drop the prefix $G$, if the graph is clear from the context.
The set of children and parents of~$v$ are~$\children(v)$ and~$\parents(v)$, respectively.
A vertex~$u$ is an \emph{ancestor} of~$v$, written~$u \in \anc(v)$, if~$G$ contains a directed, possibly empty, path from~$u$ to~$v$.
A vertex~$v \in V$ or an edge~$uv \in E$ has \emph{offspring in~$A \subseteq X$} if~$v$ is an ancestor of some vertex in~$A$.

\paragraph*{Phylogenetic diversity}

\begin{figure}[t]
	\centering\hspace{-8ex}
	\includegraphics[width=0.90\linewidth]{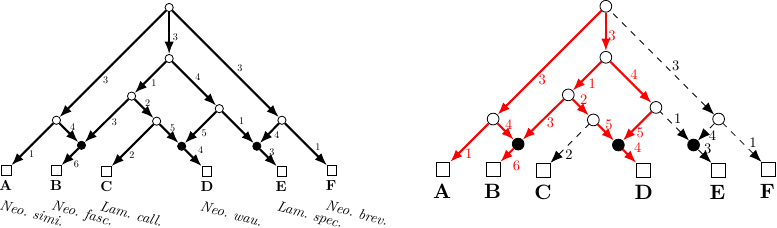}
	\caption{Left: A phylogenetic network~$N$ representing an idealized evolutionary heritage of a set of fish species~\cite{KDS+07}.
		The edge weights are included for illustrative purposes and are not taken from~\cite{KDS+07}.
		Right: A visualization of all paths leading to~$Z = \{A,B,D\}$. The value~$\MAP(Z)$ is~$41$.
	}
	\label{fig:example-PD-1}
\end{figure}

For a phylogenetic tree~$\Tree$ on $X$, the \emph{phylogenetic diversity} of a set of taxa~$A \subseteq X$, denoted by $\PD(A)$, is the total weight of all edges in~$\Tree$ with offspring in~$A$, or equivalently, the total weight of edges on paths from the root to taxa in~$A$. We consider three different extensions of phylogenetic diversity to networks. On trees, all three measures coincide with~$\PD(A)$.

A natural first extension to a phylogenetic network~$N$ on $X$ considers \emph{all} paths from the root to taxa simultaneously. The \emph{all-paths phylogenetic diversity} of~$A \subseteq X$ in~$N$, denoted by~$\MAP (A)$, is the total weight of all edges in~$N$ with offspring in~$A$ (see \Cref{fig:example-PD-1}).

A second family of PD on networks arises from switchings. Formally, letting $R(\Net) \subseteq V$ be the set of reticulations of a phylogenetic network~$\Net = (V, E, \w)$, a \emph{switching}~$\sigma: R(\Net) \to E$ is a function that maps each reticulation vertex to one of its incoming edges (called \emph{reticulation edges}). The \emph{switching tree~$\Tree_\sigma$} of~$\Net$ for a switching~$\sigma$ is the subgraph of~$\Net$ obtained by deleting all reticulation edges except for those in the image of~$\sigma$. Note that~$\Tree_\sigma$ may have chain vertices and unlabeled leaves.
Nevertheless, we extend the definition of phylogenetic diversity naturally to switching trees. 
We denote by $S(\Net)$ the set of all switching trees of~$\Net$. Since different switching trees may yield different diversity values, we obtain an optimistic and a pessimistic diversity variant. For a phylogenetic network~$N$ on $X$ and a set of taxa $A \subseteq X$, we define the \emph{max-tree phylogenetic diversity}~$\MxPD (A) \;:=\; \max_{\Tree_\sigma \in S(\Net)} \PDsub{\Tree_\sigma}(A)$ and the \emph{min-tree phylogenetic diversity}~$\MnPD (A) \;:=\; \min_{\Tree_\sigma \in S(\Net)} \PDsub{\Tree_\sigma}(A)$. See \Cref{fig:example-PD-2}.
These three definitions of phylogenetic diversity were introduced in~\cite{bordewichNetworks,WickeFischer2018}.

To model conservation under limited resources, each taxon~$x \in X$ carries an integer \emph{cost}~$c(x) \in \mathbb{N}_{\ge0}$. A \emph{budget}~$B \in \mathbb{N}_{\geq 0}$ limits the total cost of selected taxa; we set~$\B := c_\Sigma(X) - B$ and note that~$c_\Sigma(A) \leq B$ if and only if $c_\Sigma(X \setminus A) \geq \B$, a symmetry we exploit algorithmically. Optimizing PD over taxon subsets under a budget is the natural goal.

\problemdef{\MAPPDlong~(\bMAPPD)}{
    A phylogenetic network~$\Net = (V, E, \w)$ on~$X$, costs~$c \colon X \to \mathbb{N}_{\ge0}$, and a budget~$B$.
}{
    A set~$A \subseteq X$ with~$c_\Sigma(A) \leq B$ that maximizes $\MAP(A)$.
}

We define the problem \MxTPDlong~(\bMxTPD) analogously, except that the objective is to maximize $\MxPD(A)$.

For~\MnTPD, however, even computing~$\MnPD(A)$ for a fixed set~$A$ is \NP-hard~\cite{bordewichNetworks}. We therefore study the evaluation problem directly. We note that when computing the $\MnPD$-value for a predetermined set of taxa, introducing a budget parameter is unnecessary.

\problemdef{\MnTPDlong~(\MnTPD)}{
    A phylogenetic network~$\Net = (V, E, \w) $ on $X$ and a set~$A \subseteq X$.
}{
    The value~$\MnPD(A)$.
}

We follow the switching tree framework as established in \cite{AVGTree, AVGTree2}, which can be viewed as a reformulation of the display tree model used in \cite{bordewichNetworks}. Unlike display trees, switching trees may contain ``leaves'' not in the set of taxa $X$ (see \Cref{fig:example-PD-2}). However, as paths to leaves that are not in~$X$ do not contribute to any PD score, the two formulations are equivalent.

\begin{figure}[t]
	\centering\hspace{-8ex}
	\includegraphics[width=0.90\linewidth]{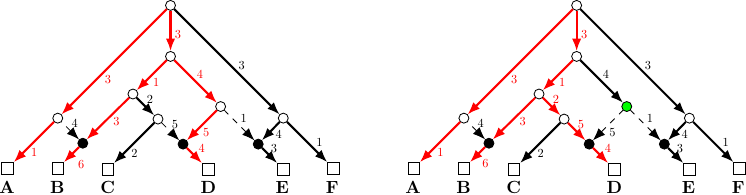}
	\caption{
		Two switching trees of the idealized network in \Cref{fig:example-PD-1} are shown, with the paths leading to~$Z = \{A,B,D\}$ highlighted.
		Left is the switching tree of maximal cost. The value~$\MxPD(Z)$ is~$30$.
		Right is the switching tree of minimal cost. The value~$\MnPD(Z)$ is~$28$.
		Observe that the vertex~\tikzdot{green} is a leaf of the switching tree, but not a taxon.
	}
	\label{fig:example-PD-2}
\end{figure}

\paragraph*{Node scanwidth}
A \emph{tree-extension} of a DAG~$G = (V, E)$ is a rooted tree~$\Gamma$ on~$V$ such that for every edge~$uv \in E$, the vertex~$u$ is an ancestor of~$v$ in~$\Gamma$. For each~$v \in V$, let~$\Gamma_{\leq v}$ denote the subtree of~$\Gamma$ rooted at~$v$. The \emph{node (scanwidth) bag}\footnote{Unlike in e.g. the standard definition of treewidth, the bag is not given explicitly but obtained implicitly.} of~$v$ in~$\Gamma$ is $\GW(v,\Gamma) \;:=\; \parents \bigl(V(\Gamma_{\leq v})\bigr) \setminus V(\Gamma_{\leq v})$, where we consider the $G$-parents of $V(\Gamma_{\leq v})$, and not in~$\Gamma$. We write~$\GW(v)$, if~$\Gamma$ is clear from context. The \emph{node scanwidth} of~$\Gamma$ is~$\nswwithoutN(\Gamma) := \max_{v \in V} |\GW(v,\Gamma)|$, and the \emph{node scanwidth} ($\nswwithoutN_G$) of~$G$, is the smallest node scanwidth of all its tree-extensions. \NSWTE is the associated optimization problem of finding an optimal tree-extension.

The \emph{(edge) scanwidth} ($\swwithoutN$) is defined analogously, except that in $\GW$, the edges between~$V(\Gamma_{\leq v})$ and~$V \setminus V(\Gamma_{\leq v})$ are counted~\cite{berry2020scanning}.
We observe that~$\nswwithoutN_G \leq \swwithoutN_G$, but it is not possible to bound~$\swwithoutN_G$ in any function in~$\nswwithoutN_G$~\cite{BruchholdThesis}.

\paragraph*{Computing dynamic programming algorithms efficiently}
In this section, we introduce a helpful tool that we use throughout the paper and which is a standard technique for dynamic programming algorithms over trees~\cite{bulteau2019parameterized,cygan}, but we formally define it here for completeness.
First, we define~$P(S)$ for a set~$S$ as the set of partitions of~$S$ (with~$t$ blocks, where~$t$ is determined later).
We define~$d(k)$ for an integer~$k$ to be the set of tuples~$(k^{(1)},\dots,k^{(t)})$, where each is an entry in~$[k]_0$, with~$\sum_{i=1}^{t} k^{(i)} = k$.

\begin{restatable}[\appendixstar]{lemma}{lemDP}
\label{lem:DP}
	Let~$\Tree$ be a tree and~$v$ a vertex with children~$w_1,\dots,w_{\outdeg(v)}$.
	Let~$\DP$ be a dynamic programming table indexed by i) the vertices of~\Tree, ii) sets~$S \subseteq M$ for some set~$M$, and iii) integers~$k$.
	For a given vertex~$v$, all values of
	\begin{equation}
		\label{eq:DP}
		\DP[v,S,k] = 
		\max_{(S_{1},\dots,S_{\outdeg(v)}) \in P(S)}
		\;\;\;
		\max_{(k_1,\dots,k_{\outdeg(v)}) \in d(k)}
		\;\;\;
		\sum_{j=1}^{{\outdeg(v)}} \DP[w_j,S_j,k_j]
	\end{equation}
	can be computed in~$\Oh(3^{|M|} \cdot  k^2 \cdot {\outdeg(v)})$~time, assuming that~$\DP[w_j,\cdot,\cdot]$ is known for all~$w_j$.
	The same is also true if the maxima are replaced with minima.
\end{restatable}
The point of this lemma is that the natural way to combine child subproblems, by enumerating all partitions of the boundary set $S$ across the children, is far too expensive. The recurrence above instead processes the children one at a time, which reduces the cost of the combination step to $O(3^{|M|})$, the number of pairs $(S', S)$ with $S' \subseteq S \subseteq M$.
The set $M$ can be chosen to be suitable for the algorithm.
In our applications, we will have~$M$ to be (a subset of) the node scanwidth bag of~$v$.

\section{Computing node scanwidth}
\label{sec:compute-nsw}

In this section, we show how to solve \NSWTE, that is, we find an optimal tree-extension together with the node scanwidth for a given DAG~$G$.
We do this in two steps; first we apply reduction rules to identify and remove easy structures in~$G$, and then we show how to solve \NSWTE on the remaining small subgraphs with an algorithm based on~\cite{holtgrefe2026exact}, or with an \ILP.

Throughout this section, let~$G = (V,E)$ be a DAG---not necessarily a phylogenetic network.
As we operate with general DAGs, reticulations in this section are defined as vertices with in-degree at least~2 and no constraint on the out-degree.
This ensures that every non-leaf, non-root vertex is either a reticulation, a tree-vertex, or a chain vertex.
We note that there might be several roots, and leaves may be reticulations.

\subsection{Preliminary observations}
We first collect five basic facts about node scanwidth that we rely on throughout the rest of the paper.

The bag of any vertex contains all its parents.
We conclude the following.
\begin{observation}
	\label{obs:in-deg}
	The node scanwidth of any DAG is at least the maximum in-degree of its vertices.
\end{observation}

For a given DAG~$G$ with tree-extension~$\Gamma$, we observe that~$\Gamma$ is also a tree-extension when removing edges from~$G$.
For each vertex~$v$, identifying~$v$ with its $\Gamma$-parent, or removing $v$ if it is the root of~$\Gamma$, creates a tree-extension of~$G-v$ that does not have larger node scanwidth.
\begin{observation}
	\label{obs:subnet}
	For any DAG~$G$, the node scanwidth does not increase when considering a (not necessarily induced) subgraph.
\end{observation}

We continue with an observation that will be particularly useful when computing node scanwidth.
\begin{observation}
	\label{obs:deg2-tree}
	Let~$v$ be a chain vertex with exactly one $G$-child~$w$.
	There is an optimal tree-extension~$\Gamma$ of~$G$ in which~$w$ is the only child of~$v$ in~$\Gamma$ and $\GW(v) \subseteq \GW(w) \cup \parents(v)$ and $\GW(w) \subseteq \GW(v) \cup \parents(w)$.
\end{observation}
\begin{proof}[Proof of \Cref{obs:deg2-tree}]
	Assume first that in a tree-extension~$v$ is directly above~$w$.
	The claims about the bags of~$v$ and~$w$ directly follow.
	
	Let~$\Gamma$ be an optimal tree-extension of~$G$.
	If~$v$ has more than one child, then attach all children that do not have a path to~$w$ to another child of~$v$.
	If~$v$ is directly above~$w$ in~$\Gamma$, then we are done.
	Assume that this is not the case.
	Since there is an arc~$vw$ in~$G$, there is a directed path~$P$ from~$v$ to~$w$ in~$\Gamma$.
	Define~$\Gamma^*$ to be the tree that results from~$\Gamma$ by removing~$v$ (and attaching all $\Gamma$-children of~$v$ to the $\Gamma$-parent of~$v$) and inserting $v$ directly above~$w$.
	We observe that~$\Gamma^*$ is a tree-extension since each directed path that exists in~$\Gamma$ also exists in~$\Gamma^*$ except for the ones from~$v$ to a vertex of~$P$, which do not correspond to edges in~$G$.
	It remains to show that the node scanwidth does not increase with the operation.
	Since~$v$ only has one child~$w$, only the bags of vertices that are between~$v$ and~$w$ are affected, that is, the vertices of~$P$.
	The bags of vertices in~$P$ now contain~$\parents(v)$, but do not contain~$v$ anymore.
	Because~$v$ is a chain vertex, its size does not increase.
\end{proof}

Lastly, we consider the following easy cases.

\begin{observation}
	\label{obs:trees}
	If~$E$ is empty, then any tree with vertices~$V$ is a tree-extension of~$G = (V,E)$ of node scanwidth~0.
	
	If $G$ is a rooted tree, then~$G$ is a tree-extension of $G$ of node scanwidth~$\min(1,|E|)$.
\end{observation}

To see that \Cref{obs:trees} is correct, we observe that in an edgeless DAG any tree-extension has an empty bag~$\GW(v)$ for every vertex~$v$ and the tree-extension of a rooted tree is simply the graph itself.

\begin{observation}
	\label{obs:1-blob}
	Let~$G$ be a bi-connected DAG with exactly one reticulation~$r$ and potentially several roots.
	Then, the node scanwidth of~$G$ is~$\indeg(r)$.
	The optimal tree-extension is a path of~$V(G)$ in any topological order.
\end{observation}
Since~$G$ is bi-connected and has only one reticulation, $r$ cannot reach any vertex but can be reached by any other vertex, and thus the tree-extension has to be a path.
Since~$r$ is the only reticulation in~$G$, for each vertex~$v$ with child~$w$ in any tree-extension, we have~$|\GW(v)| \leq |\GW(w)|$, which shows that~$\GW(r)$ is the largest bag.

\subsection{Data reduction rules}
We continue with two polynomial-time reduction rules that reduce the complexity of a DAG~$G$: \Cref{rr:deg2-chain} contracts redundant chains of degree-2 vertices, and \Cref{rr:cut-vertex} splits $G$ at cut vertices into independent subinstances. Together with the base cases above, these reduce any instance to small bi-connected components with no adjacent chain vertices.

\newcommand{\rrChainVertices}[1]{
	\begin{rr}#1
		Let~$v$ and~$w$ be chain vertices and let~$v$ be the parent of~$w$.
		Suppress~$v$.
	\end{rr}
}
\rrChainVertices{\label{rr:deg2-chain}}
\newcommand{\lemChainVertices}[1]{
	\begin{lemma}#1
		\Cref{rr:deg2-chain} is correct and can be applied exhaustively in~$\Oh(|E|)$ time.
		
		Let~$G'$ be the resulting DAG with optimal tree-extension~$\Gamma'$.
		Then, adding~$v$ directly above~$w$ in~$\Gamma'$ creates an optimal tree-extension~$\Gamma^*$ of~$G$ with the same node scanwidth.
	\end{lemma}
}
\lemChainVertices{\label{lem:deg2-chain}}
\begin{proof}
	We first observe that any tree-extension~$\Gamma$ of~$G$ is a tree-extension of~$G'$, after suppressing~$v$ in~$\Gamma$, with the same node scanwidth.
	Thus, $\nswwithoutN_{G'} \le \nswwithoutN_{G}$.
	This shows that~$\Gamma^*$ is an optimal tree-extension.
	
	Let~$\Gamma'$ be any tree-extension of~$G'$.
	Let~$u$ be the $\Gamma'$-parent of~$w$ and, by~\Cref{obs:deg2-tree}, we may assume that~$w'$ is the only $\Gamma'$-child of~$w$.
	Let~$\Gamma^*$ be the tree-extension resulting from subdividing~$uw$ and naming the resulting vertex~$v$, in~$\Gamma^*$.
	Every bag in~$\Gamma^*$ equals the respective bag in~$\Gamma'$, except the bags of~$v$ and~$w$.
	By \Cref{obs:deg2-tree}, we conclude~$\GW(v,\Gamma^*) \subseteq \GW(w,\Gamma^*) \cup \{u\}$ and $\GW(w,\Gamma^*) \subseteq \GW(v,\Gamma^*) \cup \{v\}$ and since~$v \not\in \GW(v)$ we conclude~$|\GW(w,\Gamma^*)| \ge |\GW(v,\Gamma^*)|$.
	It remains to show that the size of~$\GW(w,\Gamma^*)$ equals at most the size of one of the bags of a vertex in~$\Gamma'$.
	We observe that~$w$ has different parents in~$G$ and in~$G'$ and, consequently, $\GW(w,\Gamma^*) \cup \{u\} = \GW(w,\Gamma') \cup \{v\}$.
	If~$u \not\in \GW(w,\Gamma^*)$, then both bags have the same size and we are done.
	Otherwise, if~$u \in \GW(w,\Gamma^*)$, then~$u \in \GW({w'},\Gamma^*) = \GW({w'},\Gamma')$, since~$uw \not\in E(G)$.
	By \Cref{obs:deg2-tree}, $\GW({w'},\Gamma') \supseteq \GW({w},\Gamma')$ and also~$w \in \GW({w'},\Gamma')$.
	We conclude
	$|\GW({w},\Gamma^*)| \le |\GW({w},\Gamma')| + 1 \le |\GW({w'},\Gamma')|$.
	Thus, the node scanwidth of~$\Gamma^*$ is at most the node scanwidth of~$\Gamma'$.
	
	We iterate over~$E$ twice; to compute all chain vertices and to find connected chain vertices, each.
	Any application takes constant time, so we apply \Cref{rr:deg2-chain} exhaustively in linear time.
\end{proof}
We note that, unlike in the computation of scanwidth, \Cref{rr:deg2-chain} indeed requires two chain vertices to be connected by an edge.
We give an illustrative example in \Cref{fig:dangerous-chain-removal}.
\begin{figure}[t]
	\centering
    \includegraphics{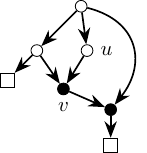}
	\caption{The node scanwidth of this example DAG is~3, as~$|\GW(v)| = 3$ in any tree-extension, but if \Cref{rr:deg2-chain} were applied to~$u$, then the node scanwidth would decrease to~$2$.}
	\label{fig:dangerous-chain-removal}
\end{figure}

\begin{rr}
	\label{rr:cut-vertex}
	We apply this reduction rule only to connected DAGs~$G$ with a single root.
	Let $v$ be a vertex of $G$. If~$v$ is a cut vertex, that is, $G - v$ has connected components~$C_1,\dots,C_\ell$,
	then solve \NSWTE independently on~$G_i := G[V(C_i) \cup \{v\}]$ for each~$i\in [\ell]$.
\end{rr}
\begin{lemma}
	\label{lem:cut-vertex}
	\Cref{rr:cut-vertex} is correct and can be applied exhaustively in~$\Oh(|E|)$~time.
	
	Let~$\Gamma_i$ be an optimal tree-extension of~$G_i$, for~$i \in [\ell]$.
	An optimal tree-extension of~$G$ is~$\Gamma := (V, \bigcup_{i\in [\ell]} E(\Gamma_i))$ with its node scanwidth being the maximum node scanwidth of any~$\Gamma_i$, for~$i\in [\ell]$.
\end{lemma}
\begin{proof}
	We first observe that~$\Gamma$ is actually a tree.
	We observe that the pairwise intersection of the vertex sets~$V(G_i)$ is always~$\{v\}$.
	Since~$G$ is connected and has only one root, there is at most one~$G_k$ in which~$v$ has incoming edges.
	Then, $v$ is the root in~$\Gamma_i$ for~$i\neq k$.
	Since there are no edges between vertices in~$C_i$ and~$C_j$ for any tuple~$i,j\in [\ell]$, $i\ne j$, $\Gamma$ is indeed a tree-extension of~$G$.
	The bound of the node scanwidth follows with the observation that~$\GW(v,\Gamma) = \GW(v,\Gamma_k)$ and~$\GW(w,\Gamma) = \GW(w,\Gamma_i)$ for any~$w \in V(C_i)$, $i\in [\ell]$.
	
	Concerning the running time, in $\Oh(m)$ time, we can find all cut-vertices and thus vertices on which the rule can be applied.
	Any application can be done in linear time.
\end{proof}

For some well-structured bi-connected components, computing node scanwidth is easy.
For example, the node scanwidth of a bi-connected component with a single reticulation is the number of incoming edges at the reticulation.
In the appendix, we present \Cref{rr:reticulation-removal}, which reduces the number of reticulations in a certain type of more complicated bi-connected components.
As initial experiments did not yield an improved running time, we only present the reduction rule in the appendix as a source of potential inspiration.

\subsection{Dealing with remaining hard instances}
Since \NSWTE is \NP-hard~\cite{BruchholdThesis}, unless P = \NP, polynomial time reduction rules will not suffice and we adapt a known algorithm for scanwidth~\cite{holtgrefe2026exact}. We deviate by omitting its reduction rules, which are partly incorrect for node scanwidth, and by modifying the computation of the bags~$\GW$ (Lines~7, 9 and 12 in Algorithm~3 of~\cite{holtgrefe2026exact}), counting nodes instead of edges. For self-containment, we sketch the algorithm below.

The algorithm utilizes a dynamic programming table over vertex subsets~$W$, computing the node scanwidth of $G[W]$ by branching into its weakly connected components. This induces a topological ordering of the vertices, yielding a tree-extension. Branches exceeding the node scanwidth threshold are pruned; see~\cite{holtgrefe2026exact} for a correctness proof.

In the appendix, we define an \ILP formulation to solve \NSWTE. This definition is not meant to be a fast algorithm, as \ILP definitions for similar problems like \textsc{Treewidth} are comparatively slower than other algorithms~\cite{dell2017first,koster2001treewidth} and also in our tests, this formulation was worse by magnitudes. The intention of this definition is twofold. Firstly, the---to the best of our knowledge---only definition of an \ILP formulation for the edge scanwidth variant of this problem~\cite{HoltgrefeThesis} has several flaws, such as insufficient criteria for ensuring that the tree-extension is indeed a tree. Secondly, we hope that this \ILP formulation may spark new ideas on how to create faster algorithms in the future.

\section{Algorithms for budgeted phylogenetic diversity on networks}
\label{sec:apply}

\begin{figure}[t]
	\centering
    \includegraphics[width=0.68\textwidth]{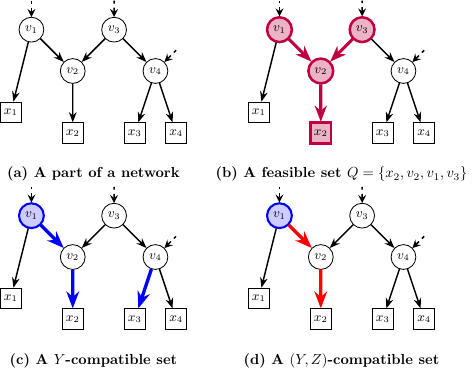}
	\caption{Visualization of the structures defined in the algorithms for a network snippet in (a); (b) shows a feasible set $Q$. For $Y=Z=\{v_1\}$, (c) shows a $Y$-compatible set and (d) shows a $(Y,Z)$-compatible set. We note that the edge $v_4x_3$ is possible in (c), but not in (d), because~$v_4 \not\in Z$.}
	\label{fig:feasible-compatible}
\end{figure}

In this section we give our algorithms for budgeted PD on networks, parameterized by node scanwidth plus the budget.
To situate the results, we first recall what is already achievable through \emph{treewidth}---a measure of tree-likeness smaller than node scanwidth~\cite{BruchholdThesis}.
Applying a meta-theorem given in~\cite{arnborg1991easy}, it follows that \bMAPPD can be solved in~$\Oh(f(\tw) \cdot \min^2(B,\B) \cdot n)$~time, where $f$ is some computable function and \tw is the treewidth of the network.
More recently, it has been shown that the unbudgeted version \MAPPD can be solved in~$\Oh^*(9^\tw)$ when given an optimal tree decomposition~\cite{MAPPD}, a result that can be extended to the budgeted version.
In this section, we consider the larger parameter node scanwidth and justify this selection with a significant improvement in the base of the running time.
More precisely, we conduct the following.

Given a phylogenetic network~\Net and a tree-extension of node scanwidth~$\nsw$, we compute and optimize phylogenetic diversity as follows.
We first show that \bMAPPD and \bMxTPD are optimized in~$\Oh(3^\nsw \cdot \min^2(B,\B) \cdot n)$~time.
The value~$\MnPD(A)$ is computed in~$\Oh(4^\nsw \cdot n)$~time, for any~$A\subseteq X$.
We restrict our focus to these PD variants, as optimizing unweighted \textsc{Network-PD} and \textsc{Average-Switching-Tree-PD} is \NP-hard even for constant edge scanwidth and therefore node scanwidth~\cite{MaxNPD,AVGTree2}.
These algorithms utilize structural definitions introduced inline, with visual examples provided in \Cref{fig:feasible-compatible}.

\subsection{Optimizing \bMAPPD and \bMxTPD}
\label{sec:algorithms-mappd-mxtpd}

We recall that~$B$ is the budget and that~\B is the total cost of all taxa minus~$B$.
The following algorithms can analogously be applied to the versions of the respective problems in which costs are uniform.
In these cases, $B$ and~\B are bounded by the instance size, such that our results imply that \MAPPD and \MxTPD are fixed-parameter tractable with respect to node scanwidth when given an optimal tree
extension.
In the case of unit costs, \B is the number of taxa that go extinct, arguably a more intuitive parameter than in the heterogeneous cost case.
Let~$\Instance = (\Net=(V,E,\w),c,B)$ be an instance of \bMAPPD or \bMxTPD, respectively.

\begin{theorem}
	\label{thm:MAPPD+MaxTree}
	Given a tree-extension with node scanwidth~$\nsw$, an optimal solution for \bMAPPD and \bMxTPD can be found in~$\Oh(3^\nsw \cdot \min^2(B,\B) \cdot n)$~time.
\end{theorem}

Before we start defining the algorithms, we need some further notation.
Let a tree-extension~$\Gamma$ of~\Net be given.
Recall that~$\Gamma_{\le v}$, for each~$v\in V$, is the subtree of~$\Gamma$ rooted at~$v$ and that~$\GW(v)$ are the vertices in~$\parents(V(\Gamma_{\le v}))$ that are not in~$\Gamma_{\le v}$.

We define the algorithms for \bMAPPD and for \bMxTPD after each other, as they are largely similar but have some differences.
For each of these problems, we define a dynamic programming algorithm with tables~$\DP_1[v,Y,b]$ and~$\DP_2[v,Y,b]$, for dimensions~$v\in V(\Net)$, $Y\subseteq \GW(v)$, and~$b\in [\min(B,\B)]_0$.
In practice, if~$B \leq \B$, then we only need table~$\DP_1$ and otherwise we only need~$\DP_2$, but here, we define both.
For brevity, we write~$\DP_i'[v,Y,b] = \max_{(Y_1,\dots,Y_{\outdeg(v)}) \in P(Y)} \max_{(b_1,\dots,b_{\outdeg(v)}) \in d(b)}  \sum_{j=1}^{{\outdeg(v)}} \DP_i[w_j,Y_j,b_j]$ in the spirit of \Cref{lem:DP}, which we use in this form in the next two algorithms.
	
\paragraph*{Algorithm for optimizing \bMAPPD (\SbMAPPD)}
To quantify table entries, we need the following definition.
A set~$Z\subseteq V$ is \emph{feasible} if each~$z\in Z$ is a leaf or has at least one \Net-child in~$Z$.
(See \Cref{fig:feasible-compatible} for an example.)
We write~$\w_\Sigma(Z)$ for~$\sum_{v\in Z} \w_\Sigma(E_v^-)$---the total weight of incoming edges at~$Z$.
We store in~$\DP_1[v,Y,b]$ the maximum value~$\w_\Sigma(Z)$ of any set~$Z \subseteq V(\Gamma_{\le v})$ such that~$Y\cup Z$ is feasible and~$c_\Sigma(Z\cap X) \le b$.
Analogously, in~$\DP_2[v,Y,b]$ we store the maximum value~$\w_\Sigma(Z)$ of any set~$Z \subseteq V(\Gamma_{\le v})$ such that~$Y\cup Z$ is feasible and~$c_\Sigma(X(\Gamma_{\le v}) \setminus Z) \ge b$.
The latter condition captures that all taxa that are not in~$Z$ have total cost of at least~$b$ and relies on the observation that~$c_\Sigma(X') \leq B$ and~$c_\Sigma(X\setminus X') \geq \B$ are equivalent.

Leaves~$x\in X$ have exactly one parent~$v_x$.
Observe that~$\GW(x) = \{v_x\}$.
We store
\begin{align}
	\label{eq:bMAPPD-base-1}
	\DP_1[x,Y,b] =~&
	\begin{cases}
		\w(v_x x) & \text{if~$b\ge c(x)$}\\
		-\infty & \text{if~$b<c(x)$ and~$Y = \{v_x\}$}\\
		0 & \text{if~$b<c(x)$ and~$Y = \emptyset$}
	\end{cases}\\
	\label{eq:bMAPPD-base-2}
	\DP_2[x,Y,b] =~&
	\begin{cases}
		\w(v_x x) & \text{if~$b = 0$,}\\
		-\infty & \text{if~$b > 0$ and~$Y = \{v_x\}$, or if~$b > c(x)$,}\\
		0 & \text{if~$0 < b\le c(x)$ and~$Y = \emptyset$.}
	\end{cases}
\end{align}

Now, let~$v$ be an internal node of \Net, either a tree-vertex or a reticulation, with incoming edges~$E_v^-$ and let~$w_1,\dots,w_{\outdeg(v)}$ be the $\Gamma$-children of~$v$.
To compute further values, we set
\begin{align}
	\label{eq:bMAPPD-rec}
	\DP_i[v,Y,b] =~&
		\max\left\{
			\DP_i'[v,Y,b];
			\w_\Sigma(E_v^-) + \DP_i'[v,(Y \setminus \parents(v)) \cup \{v\},b]
		\right\}.
\end{align}

Let~$\rho$ be the root of~\Net.
If~$B \le \B$ then the optimal $\MAP$-value is stored in~$\DP_1[\rho,\emptyset,B]$ and otherwise in~$\DP_2[\rho,\emptyset,\B]$.
Backtracking allows finding the set which maximizes~$\MAP$.

\begin{restatable}[\appendixstar]{lemma}{lemMAPPDtable}
\label{lem:MAPPD-table}
    $\DP_1[v,Y,b]$ stores the maximum value~$\w_\Sigma(Z)$ of any set~$Z \subseteq V(\Gamma_{\le v})$ such that~$Y\cup Z$ is feasible and~$c_\Sigma(Z\cap X) \le b$.
	$\DP_2[v,Y,b]$ stores the maximum value~$\w_\Sigma(Z)$ of any set~$Z \subseteq V(\Gamma_{\le v})$ such that~$Y\cup Z$ is feasible and~$c_\Sigma(X(\Gamma_{\le v}) \setminus Z) \ge b$.
\end{restatable}
\begin{proof}[Proof Sketch]
	The proof proceeds by induction over the tree-extension~$\Gamma$, with the leaf recurrences forming the base case.
	For an internal vertex~$v$, \Recc{eq:bMAPPD-rec} distinguishes whether or not~$v$ belongs to the feasible set.
	If~$v$ is excluded, the boundary requirements in~$Y$ and the available budget are distributed among the $\Gamma$-child subtrees, and the induction hypothesis yields feasible child solutions whose disjoint union has the claimed value.
	If~$v$ is included, the weight~$\omega_\Sigma(E_v^-)$ of its incoming edges is added, the parents of~$v$ are removed from the remembered set because they obtain~$v$ as a child, and~$v$ is passed to one child subproblem to ensure that~$v$ itself has a selected child.
	Consequently, every value produced by the recurrence corresponds to a feasible set satisfying the relevant cost bound.
	Conversely, every feasible set can be split among the $\Gamma$-child subtrees, and, depending on whether it contains~$v$, one of the two terms of \Recc{eq:bMAPPD-rec} attains at least its value.
	The same argument applies to~$\DP_2$ after replacing the upper bound on the cost of selected taxa with the lower bound on the cost of omitted taxa.
\end{proof}

\begin{lemma}
	\label{lem:MAPPD-correct}
	\SbMAPPD is correct.
\end{lemma}
\begin{proof}
	By \Cref{lem:MAPPD-table}, we know that the tables store the desired entries.
	Let~$\rho$ be the root of the network.
	It remains to show that, given any integer~$d$, there is a set~$S \subseteq X$ of cost at most~$B$ and with~$\MAP(S) = d$ if and only if~$\DP_1[\rho,\emptyset,B] \geq d$ if~$B \leq \B$ and~$\DP_2[\rho,\emptyset,\B] \geq d$ if~$B > \B$.
	We only show the equivalence for~$B \leq \B$ and omit the analogous other case.
	
	Let there be a set~$S \subseteq X$ of cost at most~$B$ optimizing~$\MAP(S)$.
	Define~$Z$ as the set of ancestors of~$S$.
	We observe that~$Z$ is feasible and~$c_\Sigma(X \cap Z) = c_\Sigma(S) \leq B$.
	Consequently, $\DP_1[\rho,\emptyset,B] \geq \w_\Sigma(Z) = \MAP(S)$.
	Now, if~$\DP_1[\rho,\emptyset,B] = d$, then there is a set~$Z$ with~$d = \w_\Sigma(Z) \leq \MAP (X \cap Z)$.
	Since~$c_\Sigma(X \cap Z) \leq B$, we conclude that~$X \cap Z$ is the desired set.
	This completes the proof of the correctness of the algorithm for \bMAPPD.
\end{proof}

Having established correctness, we now quantify the algorithm's efficiency.
\begin{lemma}
	\label{lem:MAPPD-rt}
	\SbMAPPD takes~$\Oh(3^\nsw  \cdot \min^2(B,\B) \cdot n )$~time.
\end{lemma}
\begin{proof}
	For each vertex~$v$, by \Cref{lem:DP}, all entries in $\DP_i'[v,\cdot,\cdot]$ can be computed in~$\Oh(3^\nsw  \cdot \min^2(B,\B) \cdot \outdeg(v))$~time, where $\outdeg(v)$ is the outdegree of~$v$ in~$\Gamma$. Since $\sum_v \outdeg(v) = n-1$, this proves the claimed running time.
\end{proof}

\paragraph*{Algorithm for optimizing \bMxTPD (\SbMxTPD)}
To prove \Cref{thm:MAPPD+MaxTree} for \bMxTPD, we use a similar idea, but it is not possible to index table entries with the same structure, as we need to ensure that for a set of taxa, we only count a switching tree (of maximum value) and not the entire spanning sub-network.

For a set of edges~$F\subseteq E$, let~$\Net[F]$ be the sub-network~$(V(F), F)$.
A set of edges~$F\subseteq E$ is \emph{$Y$-compatible} if~$F$ contains an outgoing edge of each~$y\in Y$, and~$\Net[F]$ is a directed forest whose leaves are in~$X$.
(See \Cref{fig:feasible-compatible} for an example.)
We let~$E(v)$, for some vertex~$v$, be the set of edges that contains all edges in~$\Gamma_{\le v}$ and edges between~$\Gamma_{\le v}$ and~$\GW(v)$.
We store in~$\DP_1[v,Y,b]$ the maximum value~$\w_\Sigma(F)$ of any $Y$-compatible set~$F \subseteq E(v)$ with~$c_\Sigma(V(F) \cap X) \le b$.
In~$\DP_2[v,Y,b]$, we store the maximum value~$\w_\Sigma(F)$ of any $Y$-compatible set~$F \subseteq E(v)$ with~$c_\Sigma(X(\Gamma_{\le v}) \setminus V(F)) \ge b$.
Again, the latter case captures that the cost of all taxa that are not in~$V(F)$ is at least~$b$.
	
For leaves~$x\in X$, in~$\DP_i[x,Y,b]$, we store the values as in~(\ref{eq:bMAPPD-base-1}) and~(\ref{eq:bMAPPD-base-2}) for \bMAPPD.
(This makes sense, as at the leaves the two variants make no difference.)
	
Now, let~$v$ be an internal node of \Net, either a tree-vertex or a reticulation, and let~$w_1,\dots,w_{\outdeg(v)}$ be the $\Gamma$-children of~$v$.
We use a different recurrence than in the previous algorithm. We set
\begin{align}
	\label{eq:bMxTPD-rec}
	\DP_i[v,Y,b] =~&
	\max\left\{
		\DP_i'[v,Y,b];
		\max_{p_v \in \parents(v)}
		\w(p_v v) + \DP_i'[v,(Y \setminus \{p_v\}) \cup \{v\},b]
	\right\}.
\end{align}

Let~$\rho$ be the root of~\Net.
If~$B \le \B$ then the optimal $\MxPD$-value is stored in~$\DP_1[\rho,\emptyset,B]$ and otherwise in~$\DP_2[\rho,\emptyset,\B]$.
Backtracking allows finding the set that maximizes~$\MxPD$.

\begin{restatable}[\appendixstar]{lemma}{lemMxTPDtable}
\label{lem:MxTPD-table}
    $\DP_1[v,Y,b]$ stores the maximum value~$\w_\Sigma(F)$ of any $Y$-compatible set~$F \subseteq E(v)$ with~$c_\Sigma(V(F) \cap X) \le b$.
    $\DP_2[v,Y,b]$ stores the maximum value~$\w_\Sigma(F)$ of any $Y$-compatible set~$F \subseteq E(v)$ with~$c_\Sigma(X(\Gamma_{\le v}) \setminus V(F)) \ge b$.
\end{restatable}
\begin{proof}[Proof Sketch]
	The proof follows the same induction and the same two directions as the proof of \Cref{lem:MAPPD-table}, but uses compatible edge sets instead of feasible vertex sets.
	The first term of the recurrence combines child forests without selecting an incoming edge of~$v$, whereas the second term selects one incoming edge, adds its weight, and requires~$v$ to have an outgoing edge in one of the child forests.
	The child edge sets lie in disjoint $\Gamma$-subtrees, so their union remains a directed forest with leaves in~$X$; conversely, every compatible forest can be decomposed according to whether it contains an incoming edge of~$v$.
	The argument for~$\DP_2$ differs only in its treatment of the cost bound.
\end{proof}

\begin{lemma}
	\label{lem:MxTPD-correct}
	\SbMxTPD is correct.
\end{lemma}
\begin{proof}
	By \Cref{lem:MxTPD-table}, we know that the tables store the desired entries.
	Let~$\rho$ be the root of the network.
	It remains to show that there is a set~$S \subseteq X$ of cost at most~$B$ and~$\MxPD(S) = d$ if and only if~$\DP_1[\rho,\emptyset,B] \geq d$ if~$B \leq \B$ and~$\DP_2[\rho,\emptyset,\B] \geq d$ if~$B > \B$.
	We only show the equivalence for~$B \leq \B$ and omit the analogous other case.
	
	Let there be a set~$S \subseteq X$ of cost at most~$B$ and $\MxPD
	(S) = d$.
	Thus, there is a rooted tree~$\Net[F]$ connecting the root and~$S$ with~$\w_\Sigma(F) = d$.
	Then, $c_\Sigma(X\cap V(F)) = c_\Sigma(S) \leq B$ and~$\DP_1[\rho,\emptyset,B] \geq \w_\Sigma(F) = d$.
	Now, if~$\DP_1[\rho,\emptyset,B] = d$, then there is an~$\emptyset$-compatible set~$F$ with~$d = \w_\Sigma(F)$ and~$c_\Sigma(V(F) \cap X) \leq B$.
	As~$\Net[F]$ is a tree, we conclude that~$V(F) \cap X$ is the desired set.
	This completes the proof of the correctness of \SbMxTPD.
\end{proof}

The running time of \SbMxTPD can be shown equivalently to \Cref{lem:MAPPD-rt}.
This completes the correctness of \Cref{thm:MAPPD+MaxTree}.

\subsection{Computing values of \MnTPD}
\label{sec:algorithms-mntpd}

In the following, we analyze \MnTPD.
We do not believe that the budgeted optimization variant of \MnTPD can be solved in~$\Oh^*(c^{\nsw} \cdot B^d)$~time for constants~$c$ and~$d$. Therefore, we only show that $\MnPD$-values can be computed in~$\Oh^*(4^\nsw)$~time, which is still a reasonable result as computing $\MnPD$-values is \NP-hard in general~\cite{bordewichNetworks}.

\begin{theorem}
	\label{thm:MinTreePD}
	Given a phylogenetic network~$\Net = (V,E, \w)$ on $X$ and a tree-extension~$\Gamma$ with node scanwidth~$\nsw$,
	the value~$\MnPD(A)$ can be computed in~$\Oh(4^\nsw \cdot n)$~time, for any~$A\subseteq X$.
\end{theorem}
\looseness=-1
To prove the result, we define a dynamic programming algorithm over~$\Gamma$.
At the core, it follows the idea of \SbMxTPD, but it uses a more restrictive mathematical object.

Let a network~$\Net = (V,E, \w)$ and a set~$A\subseteq X$ be given.
Remove all vertices which do not have offspring in~$A$ with their incident edges.
We can thus assume that~$A=X$.

We can get a tree-extension for this new network by doing the following for each vertex~$v$ that does not have offspring in~$A$.
Let~$p_v$ be the parent of~$v$ and~$w_1,\dots,w_{\outdeg(v)}$ the $\Gamma$-children of~$v$, possibly with~${\outdeg(v)}=0$.
We add edges~$p_v w_i$ for each~$i \in [{\outdeg(v)}]$ and remove~$v$ with the incident edges.
This tree-extension may not be optimal, but has node scanwidth at most~$\nsw$ and is thus sufficient for us.

We recall that~$\Net[F]$ is the sub-network~$(V(F), F)$, for a set of edges~$F\subseteq E$ and that~$E(v)$, for some vertex~$v$, is the set of edges in~$\Gamma_{\le v}$ and edges between~$\Gamma_{\le v}$ and~$\GW(v)$.	

\paragraph*{Algorithm for computing $\MnPD$~(\SbMnTPD)}
We define a dynamic programming algorithm with table entries~$\DP[v,Y,Z]$, for dimensions~$v\in V(\Net)$ and $Y \subseteq Z \subseteq \GW(v)$.
A set of edges~$F\subseteq E$ is \emph{$(Y,Z)$-compatible} if~$F$ is non-empty, $F$~contains an outgoing edge of~$y$, for each~$y\in Y$, and~$\Net[F]$ is a directed forest whose set of leaves \emph{equals} $X(\Gamma_{\leq v})$ and whose roots are in~$Z$---not necessarily covering all of~$Z$.
(See \Cref{fig:feasible-compatible} for an example.)
We store in~$\DP[v,Y,Z]$ the \emph{minimum} value~$\w_\Sigma(F)$ of any $(Y,Z)$-compatible set~$F \subseteq E(v)$.
	
Leaves~$x\in X$ have exactly one parent~$v_x$.
Observe that~$\GW(x) = \{v_x\}$.
We store
\begin{align}
	\label{eq:bMnTPD-base}
	\DP[x,Y,Z] =~&
	\begin{cases}
		\w(v_x x) & \text{if~$Z = \{v_x\}$, (Independent of~$Y$.)}\\
		\infty & \text{else, if~$Y = Z = \emptyset$.}
	\end{cases}
\end{align}
	
Now, let~$v$ be an internal node of \Net, either a tree-vertex or a reticulation, and let~$w_1,\dots,w_{\outdeg(v)}$ be the $\Gamma$-children of~$v$.
In the spirit of \Cref{lem:DP} (with the integer~$k=0$), we write~$\DP'[v,Y,Z] = \min_{(Y_1,\dots,Y_{\outdeg(v)}) \in P(Y)} \sum_{j=1}^{{\outdeg(v)}} \DP[w_j,Y_j,Z\cap\GW(w_j)]$ for brevity.
We note that this definition slightly differs from the definition of~$\DP_i'$ in the other two algorithms.
To compute further values, we set
\begin{align}
	\label{eq:bMnTPD-rec}
	\DP[v,Y,Z] =~&
	\min\left\{
	\DP'[v,Y,Z];
	\min_{p_v \in Z}
	\w(p_v v) + \DP'[v,(Y \setminus \{p_v\}) \cup \{v\},Z \cup \{v\}]
	\right\}.
\end{align}
Here, we set~$\w(p_v v) = \infty$ whenever~$p_v$ is not a parent of~$v$.
Therefore, we can equivalently minimize over~$p_v \in Z \cap \parents(v)$.
We observe that~$p_v$ is not necessarily in~$Y$.

Let~$\rho$ be the root of~\Net.
The optimal $\MnPD$-value is stored in~$\DP[\rho,\emptyset,\{\rho\}]$.
(We allow~$Z = \{\rho\}$ at~$\rho$, though~$\GW(\rho)$ is actually empty, in order to allow~$\rho$ to be the root.)

\begin{restatable}[\appendixstar]{lemma}{lemMnTPDtable}
\label{lem:MnTPD-table}
    $\DP[v,Y,Z]$ stores the minimum value~$\w_\Sigma(F)$ of any $(Y,Z)$-compatible set~$F \subseteq E(v)$.
\end{restatable}
\begin{proof}[Proof Sketch]
	The proof again follows the induction used for \Cref{lem:MAPPD-table}, but minimizes over $(Y,Z)$-compatible forests and uses~$Z$ to record their permitted roots.
	If no incoming edge of~$v$ is selected, the child forests retain their roots in~$Z$, whereas selecting an edge from~$p_v$ to~$v$ connects components rooted at~$v$ to~$p_v$ and therefore passes~$v$ as a temporary permitted root to the child subproblems.
	Conversely, every compatible forest is decomposed according to whether it contains an incoming edge of~$v$, and the corresponding term of \Recc{eq:bMnTPD-rec} has weight no larger.
	The leaf recurrence, which forces the unique incoming edge of a leaf, supplies the base case.
\end{proof}

\begin{lemma}
	\label{lem:MnTPD-correct}
	\SbMnTPD is correct.
\end{lemma}
\begin{proof}
	Assume that~$\DP[\rho,\emptyset,\{\rho\}] = d \in \mathbb{N}$.
	Then, each~$(\emptyset,\{\rho\})$-compatible set~$F \subseteq E(\rho) = E(\Net)$ satisfies~$\w_\Sigma(F) \geq d$.
	As the set of edges of each switching tree of~$\Net$ is~$(\emptyset,\{\rho\})$-compatible, we conclude~$\MnPD(X) \ge d$.
	To see that~$\MnPD(X) = d$ implies~$\DP[\rho,\emptyset,\{\rho\}] = d$, it is sufficient to observe that the family of~$(\emptyset,\{\rho\})$-compatible sets is the family of sets of edges of all switching trees of~\Net.
	This shows correctness.
\end{proof}

We conclude this section by deriving the total running time.
\begin{lemma}
	\label{lem:MnTPD-rt}
	\SbMnTPD takes~$\Oh(4^\nsw \cdot n)$~time.
\end{lemma}
\begin{proof}
    To get the claimed running time, we do the following for each internal node~$v$.
	We iterate over all subsets~$Z$ of~$\GW(v)$.
	For a given~$Z$, we compute all values of~$\DP[v,\cdot,Z]$ as in~\Recc{eq:bMnTPD-rec} with the help of~\Cref{lem:DP} for~$k=0$.
	This takes~$\Oh(3^{|Z|} \cdot {\outdeg(v)})$~time, for each vertex~$v$ and each~$Z \subseteq \GW(v)$, where $\outdeg(v)$ is the outdegree of~$v$ in~$\Gamma$.
	Since~$\sum_{A \subseteq B} 3^{|A|} = 4^{|B|}$ for any set~$B$ (a standard consequence of the binomial theorem; see, e.g., \cite{graham}) and since $\sum_v \outdeg(v) = n-1$, we conclude that the overall running time is~$\Oh(4^\nsw \cdot n)$ to compute all table entries.
\end{proof}

\section{Software and experiments}
\label{sec:software}

Having established the algorithms and their running-time guarantees, we now turn to their practical behavior on realistic networks.
 
\begin{figure}[t]
	\centering
	\includegraphics[width=0.99\textwidth]{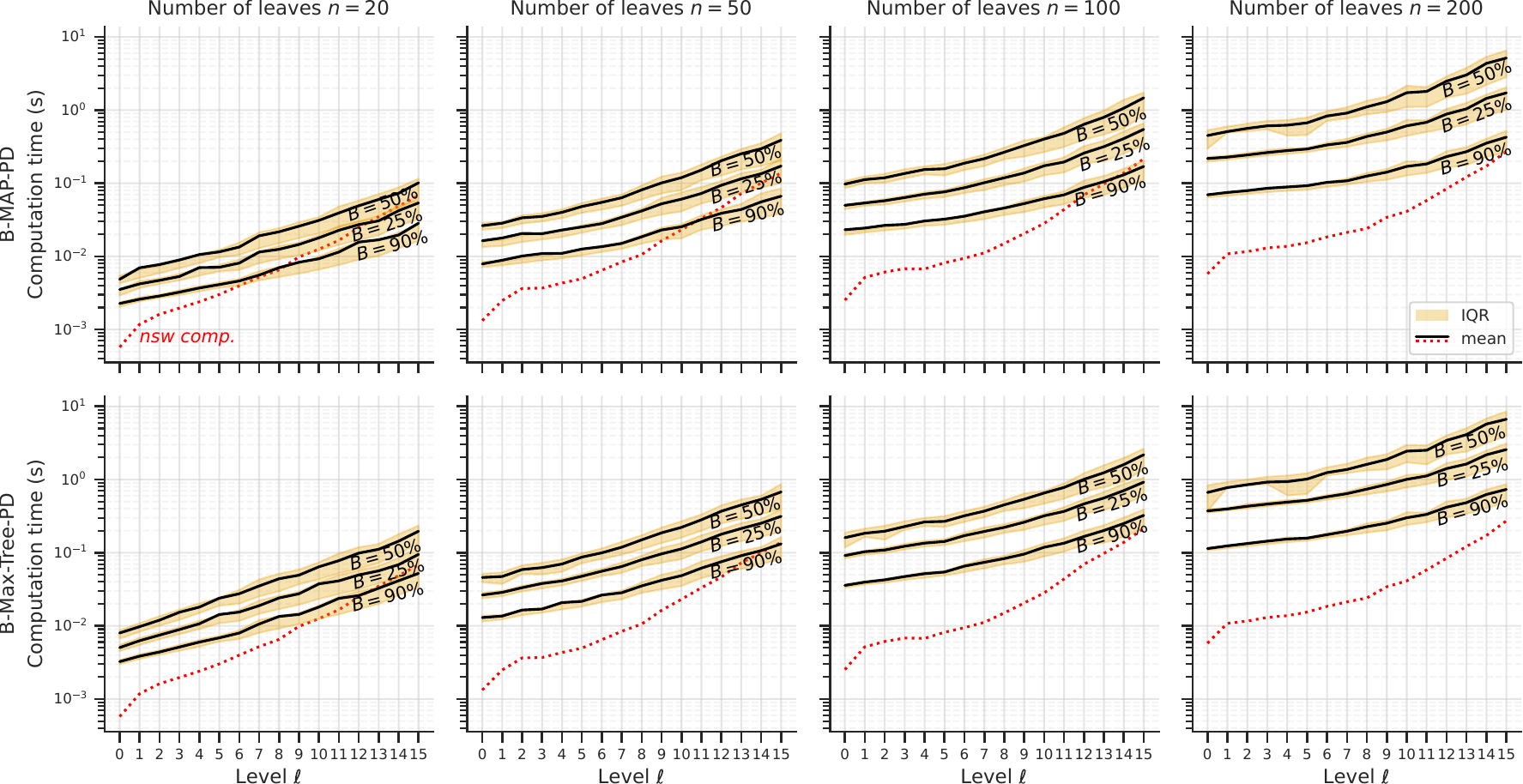}
	\caption{Computation times of \SbMAPPD and \SbMxTPD for sets of 100 $n$-leaf level-$\ell$ networks, using budgets~$B$ of $25\%$, $50\%$, and $90\%$ of the total cost. Both the mean and interquartile range (IQR) are depicted. The dotted line annotated `nsw comp.' shows the mean time to compute an optimal tree-extension; this time is not included in the other measurements.}
	\label{fig:mappd_maxtree}
\end{figure}

\paragraph*{Implementation}

We implemented our new node scanwidth algorithms and reduction rules within the existing Python package \scanwidth~\cite{holtgrefe2026exact}, which previously supported only edge scanwidth algorithms, thereby making the node scanwidth framework available as a standalone tool.\footnote{\scanwidth was also extended with heuristics for node scanwidth, adapted from scanwidth heuristics~\cite{holtgrefe2026exact}.}

Our three new phylogenetic diversity algorithms were implemented within the Python package \panda~\cite{PaNDA}, which originally supported only an edge scanwidth-based algorithm for computing all-paths diversity on phylogenetic networks.\footnote{\panda was also substantially extended with improved documentation, a more modular and extensible architecture, and integration with \phylozoo~\cite{Holtgrefe2026PhyloZoo} for increased interoperability.}

\paragraph*{Experiments}
\textbf{Simulated data.}\quad To assess computational efficiency of our algorithms, we conducted a simulation study on the effect of different network parameters on running time, using the same benchmark networks as in~\cite{PaNDA}. The dataset consists of $6{,}400$ phylogenetic networks simulated under a birth-death-hybridization model using the R package \textsf{SiPhyNetwork}~\cite{justison2023siphynetwork}. It contains $100$ level-$\ell$ networks with $n$ leaves for each combination of $n \in \{20, 50, 100, 200\}$ and $\ell \in \{0, \ldots, 15\}$. Here, the \emph{(edge-)level} is the most commonly used measure of tree-likeness in the phylogenetics literature, defined as the maximum, over all bi-connected components, of the number of reticulation edges minus the number of reticulation vertices.
Since node scanwidth is expected to be particularly advantageous over edge scanwidth for networks containing vertices of larger outdegree, we modified the networks to be non-binary by contracting the shortest 10\% of edges, while preserving adherence to our network definition.

\medskip

\noindent
\textbf{Experimental results.}\quad
Initial experiments confirmed that our dynamic programming algorithm---combined with the reduction rules---computes the node scanwidth of all considered instances with ease. Since it closely follows the edge scanwidth framework already evaluated extensively in~\cite{holtgrefe2026exact}, we focus here on the full pipeline for computing phylogenetic diversity.

To evaluate \SbMAPPD and \SbMxTPD, we assigned each taxon an integer cost sampled from a log-normal distribution with parameters $\mu = 2$ and $\sigma = 0.8$. This allows the realistic model of most taxa being of low cost, but some being quite expensive.
We then ran both algorithms to compute optimal phylogenetic diversity values under budgets $B$ equal to $25\%$, $50\%$, and $90\%$ of the total cost.
\Cref{fig:mappd_maxtree} depicts the resulting computation times, with the horizontal axis representing the network level.
In \Cref{fig:mintree} we evaluated \SbMnTPD.
Since this problem has no budget parameter, the only choice is the set of taxa $A$, which we set to $X$. Other choices are faster, as restricting to a subset of taxa effectively reduces the network size.

\begin{figure}[t]
	\centering
	\includegraphics[width=0.99\textwidth]{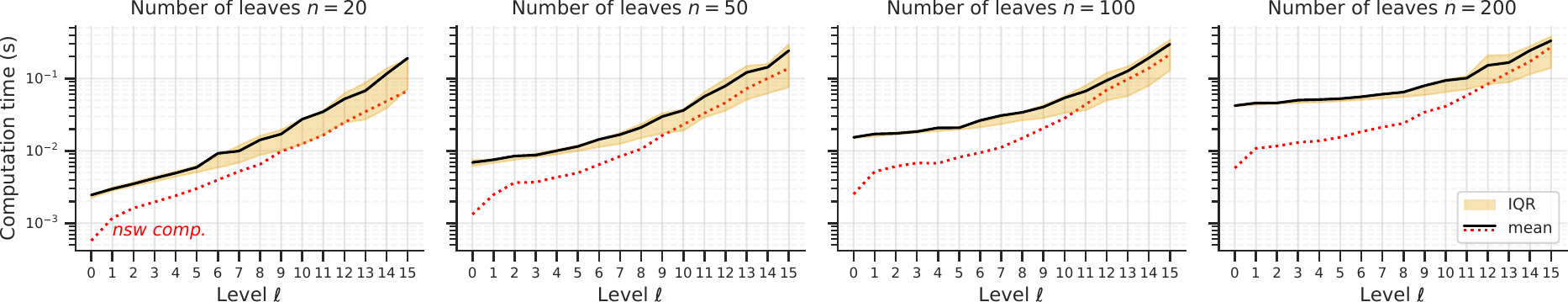}
	\caption{Computation times of \SbMnTPD for sets of 100 $n$-leaf level-$\ell$ networks, for the case where $A = X$. Both the mean and interquartile range (IQR) are depicted. The dotted line annotated `nsw comp.' shows the mean time to compute an optimal tree-extension; this time is not included in the other measurements.}
	\label{fig:mintree}
\end{figure}

Our implementation exhibits strong scalability, solving even the largest instances within seconds. As expected, for the two budgeted problems, the 50\% budget yields the longest running times, while the 90\% budget is the fastest, since our algorithm is parameterized by $\min(B,\B)$, where $\B$ denotes the complementary budget.

We further note that our new implementation substantially improves upon the previous edge scanwidth-based algorithm in \panda~\cite{PaNDA}, which required around 100 seconds on the largest unit-cost networks. While part of the gain is due to the smaller parameter \nsw, the implementation itself also contributes significantly to the observed speedup. Moreover, the node scanwidth algorithm exhibits running times comparable to the edge scanwidth algorithm of~\cite{holtgrefe2026exact} and does not appear to be the computational bottleneck on larger instances.

Finally, we performed additional large-scale evaluations on 20 networks with one thousand taxa, obtained by identifying the roots of five randomly selected 200-taxon networks from our dataset.
For the worst-case setting, where the budget is 50\%, this yielded average running times of 126 seconds for \bMAPPD and 138 seconds for \bMxTPD, while the \MnTPD computation required on average 1.1 seconds when evaluating the worst-case setting $A = X$. Together, these results showcase that our approaches scale far beyond the network sizes currently feasible for phylogenetic network inference methods.

\section{Conclusions}
\label{sec:discussion}

This work set out to reconcile two things that usually pull against each other: models of phylogenetic diversity faithful enough to capture reticulate evolution with non-uniform protection costs, and running times that remain tractable at biologically relevant scales. Node scanwidth is what makes the reconciliation possible.

While previous approaches often relied on the reticulation number~\cite{MAPPD,MaxNPD}, level~\cite{AVGTree}, or scanwidth~\cite{PaNDA,AVGTree2}---parameters that can be prohibitively larger---node scanwidth offers a smaller measure of a network's tree-likeness.
Treewidth is smaller still but yields considerably worse bases in the exponential part of the running time and less transparent algorithms~\cite{MAPPD,arnborg1991easy}. Node scanwidth thus occupies a useful middle ground: small enough to be practical on the networks arising in empirical biological datasets~\cite{holtgrefe2026exact}, yet structured enough to admit a clean dynamic-programming formulation with a small base.

A relevant observation arising from our $\Oh^*(3^{\nsw} \cdot \min^2(B,\B))$ solvers is that the addition of budgetary constraints does not lead to an explosion in the running time. Our implementation validates these results and establishes a new baseline for efficiency, beating the state-of-the-art~\cite{PaNDA} on unitary instances while additionally handling arbitrary non-binary network topologies and non-uniform conservation costs.

Several theoretical questions remain. First, we wonder whether \bMxTPD{} can be solved in~$\Oh^*(2^B)$~time, which would give a budget-only parameterization to complement our combined one.
Second, we suspect that \MnTPD is fundamentally harder than the other two variants and its budgeted optimization variant cannot be solved in~$\Oh^*(c^{\nsw} \cdot B^d)$~time for any constants~$c$ and~$d$. The intuition is that minimizing over switchings introduces a quantifier alternation absent from the other two problems: one must certify that no switching tree does better, rather than exhibit one that does well. This same alternation is why we further ask whether the budgeted optimization version of \MnTPD is hard with respect to $\Sigma^2_p$, the second level of the polynomial hierarchy~\cite{russell1998symmetric}. If the $\Oh^*(c^{\nsw} \cdot B^d)$ barrier is real, it would be interesting to investigate whether \MnTPD can at least be solved in~$\Oh^*(c^{\ell} \cdot B^d)$~time for constants~$c$ and~$d$, where $\ell$ is the \emph{vertex-level} of the network: the largest number of reticulation vertices in any of its bi-connected components.
This is immediate for binary networks but appears significantly more difficult for arbitrary networks.

Beyond these specific questions, the broader message is that the right structural parameter can make a realistically rich conservation model computationally tractable. By handling highly reticulated, non-binary networks with heterogeneous costs within a single framework, and by validating this framework on instances far larger than current network-inference methods produce, we hope to have brought phylogenetic diversity on networks one step closer to becoming a practical tool for protecting vanishing biodiversity.
\newpage

\setcounter{page}{1}
\renewcommand\thepage{\roman{page}}

\section*{Declarations}
    
\subsection*{Ethics approval and consent to participate}
Not applicable.
 
\subsection*{Consent for publication}
Not applicable.
 
\subsection*{Availability of data and materials}
The phylogenetic diversity software and the data to reproduce the experiments are available at \url{https://github.com/nholtgrefe/panda}. The node scanwidth software is available at \url{https://github.com/nholtgrefe/scanwidth}.
 
\subsection*{Competing interests}
The authors declare that they have no competing interests.
 
\subsection*{Funding}
NH was supported by the Dutch Research Council (NWO), grant OCENW.M.21.306. JS was partially supported by NWO grant~OCENW.GROOT.2019.015 (``Optimization for and with Machine Learning'') and ERC Advanced Grant~101199930 (``New Horizons of Parameterized Complexity'').
 
\subsection*{Authors' contributions}
Both authors contributed equally to the conception of the study, the development of the algorithms and theory, and the writing of the manuscript. NH implemented the algorithms and performed the computational experiments. Both authors read and approved the final manuscript.
 
\subsection*{Acknowledgements}
Generative AI was used to assist with programming and with language editing of the manuscript. The authors remain fully responsible for the content of the manuscript.

We thank anonymous reviewers of this paper for helpful comments.

\bibliography{ref}
\newpage

\setcounter{page}{1}
\renewcommand\thepage{A.\arabic{page}}

\appendix
 
\section{Omitted proofs}
\label{sec:appendix-proofs}
 
\subsection{Proof of \Cref{lem:DP}}
\lemDP*
\begin{proof}
	We show the correctness for the maximum case and omit the analogous minimum case.
	The solution to prevent considering a partition of ${\outdeg(v)}$ parts, which would be too costly, is to create an auxiliary table~$\DP'$ with another dimension~$j \in [{\outdeg(v)}]$.
	We set
	\begin{equation}
		\DP'[1,v,S,k]
		=
		\DP[w_1,S,k]\\
	\end{equation}
	and for further tables we use the recurrence
	\begin{equation}
		\DP'[j+1,v,S,k]
		=
		\max_{S' \subseteq S}
		\max_{k' \leq k}
		\DP[w_{j+1},S',k']
		+
		\DP'[j,v,S \setminus S',k-k'].
	\end{equation}
	Finally, we set
	\begin{equation}
		\DP[v,S,k]
		=
		\DP'[{\outdeg(v)},v,S,k].
	\end{equation}
	
	We can see correctness by observing that there is a partition~$(S_1,\dots,S_{\outdeg(v)}) \in P(S)$ if and only if we can select~$S_{j+1} = S'$ in the recursion for each~$j \in [{\outdeg(v)}-1]$.
	An analogous statement holds for the values of~$k$.

    The number of pairs $(S, S')$ with $S' \subseteq S \subseteq M$ is $3^{|M|}$---each item can be in $S'$, $S \setminus S'$, or in $M \setminus S$---and for each such pair we consider at most $k^2$ pairs $(k, k')$, for each of the $\outdeg(v)$ steps of the recurrence. This gives the claimed running time.
\end{proof}
We note that this approach can easily be generalized to more than one set and more than one integer.

\subsection{Proof of \Cref{lem:MAPPD-table}}
\lemMAPPDtable*
\begin{proof}
	We only show the statement for~$\DP_1$ and omit the analogous other case. Without loss of generality, let~$\DP_1[v,Y,b]$ store~$d$.
	We assume as an induction hypothesis that~$\DP_1[w_j,\cdot,\cdot]$ stores the correct value for each~$j \in [{\outdeg(v)}]$.
	We show that there is a set~$Z \subseteq V(\Gamma_{\le v})$ such that~$Y\cup Z$ is feasible, $c_\Sigma(Z\cap X) \le b$, and~$\w_\Sigma(Z) = d$.
	Afterward, we show that every set~$Z \subseteq V(\Gamma_{\le v})$ with~$Y\cup Z$ being feasible and $c_\Sigma(Z\cap X) \le b$ satisfies~$\w_\Sigma(Z) \le d$.
	
	Since~$\DP_1[v,Y,b] = d$, by~\Cref{eq:bMAPPD-rec}, $\DP_1'[v,Y,b] = d$ or~$\DP_1'[v,Y \setminus \parents(v) \cup \{v\},b] = d - \w_\Sigma(E_v^-)$.
	In the former case,
	$$
	\max_{(Y_1,\dots,Y_{\outdeg(v)}) \in P(Y)} \max_{(b_1,\dots,b_{\outdeg(v)}) \in d(b)} \sum_{i=1}^{{\outdeg(v)}} \DP_1[w_i,Y_i,b_i] = d
	$$
	and thus there are~$(Y_1,\dots,Y_{\outdeg(v)}) \in P(Y)$ and~$(b_1,\dots,b_{\outdeg(v)}) \in d(b)$ such\lb that~$\sum_{i=1}^{{\outdeg(v)}} \DP_1[w_i,Y_i,b_i] = d$.
	By the induction hypothesis, this implies that there are sets~$Z_i \subseteq V(\Gamma_{\le w_i})$ such that~$Y_i \cup Z_i$ is feasible and $c_\Sigma(Z_i\cap X) \le b_i$, for each~$i \in [{\outdeg(v)}]$.
	Further, $\w_\Sigma(Z_i) = d_i$ for some integer~$d_i$ and~$\sum_{i=1}^{\outdeg(v)} d_i = d$.
	We define~$Z$ as the (disjoint) union of these~$Z_i$.
	It follows~$c_\Sigma(Z \cap X) = c_\Sigma(X \cap \biguplus_{i=1}^{\outdeg(v)} Z_i) = \sum_{i=1}^{\outdeg(v)} c_\Sigma(X \cap Z_i) \leq \sum_{i=1}^{\outdeg(v)} b_i = b$.
	Further, since every~$Z_i$ is feasible and~$Y_i$ are a partition of~$Y$, also the union~$Z$ is feasible.
	And by definition~$\w_\Sigma(Z) = \sum_{i=1}^{\outdeg(v)} \w_\Sigma(Z_i) = \sum_{i=1}^{\outdeg(v)} d_i = d$.
	Thus, $Z$ is the desired set.
	Now, in the latter case,
	$$
	\max_{(Y_1,\dots,Y_{\outdeg(v)}) \in P(Y \setminus \parents(v) \cup \{v\})} \max_{(b_1,\dots,b_{\outdeg(v)}) \in d(b)} \sum_{i=1}^{{\outdeg(v)}} \DP_1[w_i,Y_i,b_i] = d.
	$$
	Again, by the induction hypothesis there are sets~$Z_i \subseteq V(\Gamma_{\le w_i})$ such that~$Y_i \cup Z_i$ is feasible and $c_\Sigma(Z_i\cap X) \le b_i$, for each~$i \in [{\outdeg(v)}]$.
	Further, $\w_\Sigma(Z_i) = d_i$ for some integer~$d_i$ and~$\sum_{i=1}^{\outdeg(v)} d_i = d - \w_\Sigma(E_v^-)$.
	We define~$Z := \{v\} \cup \biguplus_{i=1}^{\outdeg(v)} Z_i$.
	As~$v \not\in X$, we conclude~$c_\Sigma(Z \cap X) \leq b$ as before.
	Further, $\w_\Sigma(Z) = \w_\Sigma(E_v^-) + \sum_{i=1}^{\outdeg(v)} \w_\Sigma(Z_i) = d$.
	It remains to show that~$Y \cup Z$ is feasible.
	For each vertex in~$Y_i \cup Z_i$ for some~$i$, this follows with the feasibility of~$Y_i \cup Z_i$.
	We observe that~$v \in Y_i$ for some~$i$, as~$(Y_1,\dots,Y_{\outdeg(v)})$ is a partition of~$P(Y \setminus \parents(v) \cup \{v\})$.
	The feasibility of~$Y \cup Z$ now follows with the fact that every vertex in~$\parents(v)$ has~$v$ as a child and~$v \in Z$.
	Thus, $Z$ is the desired set, and this direction of the proof is complete.
	
	Now, let~$Z \subseteq V(\Gamma_{\le v})$ satisfy that~$Y\cup Z$ is feasible and $c_\Sigma(Z\cap X) \le b$.
	We show~$\w_\Sigma(Z) \le d$.
	We distinguish two cases, $v \in Z$ or $v \not\in Z$.
	We start with the latter, that is~$v \not\in Z$.
	Define~$Z_i := Z \cap V(\Gamma_{\leq w_i})$, for each~$i \in [{\outdeg(v)}]$.
	Since there are no edges between~$\Gamma_{\leq w_i}$ and~$\Gamma_{\leq w_j}$ for any pair~$i\neq j$, the sets~$Z_i$ are feasible.
	We define~$Y_i$ to be the subset of vertices in~$Y$ which have a child in~$Z_i$, but no child in~$Z_j$ for any~$j<i$.
	By definition of these sets, we conclude~$Y_i \cup Z_i$ is feasible.
	We define~$b_i = c_\Sigma(X \cap Z_i)$ and conclude~$\w_\Sigma(Z_i) \leq \DP_1[w_i,Y_i,b_i]$.
	Since~$(Y_1,\dots,Y_{\outdeg(v)})$ is a partition of~$Y$ and~$(b_1,\dots,b_{\outdeg(v)}) \in d(b)$, we know by~\Cref{eq:bMAPPD-rec} that~$d = \DP_1[v,Y,b] \geq \DP_1'[v,Y,b] \geq \sum_{i=1}^{{\outdeg(v)}} \DP_1[w_i,Y_i,b_i] \geq \sum_{i=1}^{{\outdeg(v)}} \w_\Sigma(Z_i) = \w_\Sigma(Z)$,
	which is what we intended to show.
	It remains to consider the case~$v \in Z$.
	We define~$Z_i$, $Y_i$ and~$b_i$ analogous to before.
	Because~$Y \cup Z$ is feasible and by the definition of a tree-extension~$v$ does not have children in~$Y$, we know that~$v$ has a child in~$Z$.
	Without loss of generality, $v$ has a child in~$Z_1$.
	We remove~$\parents(v)$ from all~$Y_i$ and add~$v$ to~$Y_1$.
	Now, $Y_i \cup Z_i$ is still feasible for each~$i \in [{\outdeg(v)}]$ and~$(Y_1,\dots,Y_{\outdeg(v)}) \in P(Y \setminus \parents(v) \cup \{v\})$.
	To show~$d \geq \w_\Sigma(Z)$, we can now use the equations~$d = \DP_1[v,Y,b] \geq \w_\Sigma(E_v^-) + \DP_1'[v,Y \setminus \parents(v) \cup \{v\},b] \geq \sum_{i=1}^{{\outdeg(v)}} \DP_1[w_i,Y_i,b_i] \geq \sum_{i=1}^{{\outdeg(v)}} \w_\Sigma(Z_i) = \w_\Sigma(Z)$.
\end{proof}
 
\subsection{Proof of \Cref{lem:MxTPD-table}}
\lemMxTPDtable*
\begin{proof}
	We only show the statement for~$\DP_1$ and omit the analogous other case.
	As an induction hypothesis, we use that~$\DP_1[w_i,\cdot,\cdot]$ stores the correct value for each~$w_i$.
	
	Let~$\DP_1[v,Y,b]$ store~$d$.
	\Cref{eq:bMxTPD-rec} implies that~$\DP_1'[v,Y,b] = d$ or~$\DP_1'[v,Y\setminus\{p_v\}\cup\{v\},b] = d-\w(p_v v)$ for some~$p_v \in \parents(v)$.
	Assume the former, then
	$$
	\max_{(Y_1,\dots,Y_{\outdeg(v)}) \in P(Y)} \max_{(b_1,\dots,b_{\outdeg(v)}) \in d(b)} \sum_{i=1}^{{\outdeg(v)}} \DP_1[w_i,Y_i,b_i] = d
	$$
	and thus there are~$(Y_1,\dots,Y_{\outdeg(v)}) \in P(Y)$ and~$(b_1,\dots,b_{\outdeg(v)}) \in d(b)$ such\lb that~$\sum_{i=1}^{{\outdeg(v)}} \DP_1[w_i,Y_i,b_i] = d$.
	By the induction hypothesis, this implies that there\lb are~$Y_i$-compatible sets of edges~$F_i \subseteq E(w_i)$ with~$c_\Sigma(V(F_i) \cap X) \leq  b_i$\lb and~$\sum_{i=1}^{{\outdeg(v)}} \w_\Sigma(F_i) = d$.
	We define~$F$ as the disjoint union of these~$F_i$ and directly receive the bound on the cost and diversity.
	To see that~$F$ is~$Y$-compatible, we observe that~$V(\Gamma_{\leq w_i})$ are pairwise disjoint and so the forests~$\Net[F_i]$ have at most the tails of the edges in~$\GW(w_i)$ in common.
	Thus, $\Net[F]$ is a directed forest whose leaves are in~$X$ and so~$F$ is a desired set.
	Now, let $\DP_1'[v,Y\setminus\{p_v\}\cup\{v\},b] = d-\w(p_v v)$ for some~$p_v \in \parents(v)$.
	Let~$(Y_1,\dots,Y_{\outdeg(v)})$ be a partition of~$Y\setminus\{p_v\}\cup\{v\}$ and let~$F_i$ be a~$Y_i$-compatible set of edges with~$c_\Sigma(V(F_i) \cap X) \leq b_i$.
	Define~$F$ as~$\{p_v v\} \cup \biguplus_{i=1}^{\outdeg(v)} F_i$.
	The bound on the cost and diversity follows as before.
	Without loss of generality, let~$v \in Y_1$.
	Since~$F_1$ is~$Y_1$ compatible, $F_1$ contains an outgoing edge of~$v$, and therefore also~$F$.
	We see with analogous arguments as before that~$\Net[F]$ is a forest with leaves in~$X$.
	Thus, $F$ is the desired set.
	
	Now, let~$F$ be a~$Y$-compatible set of edges for which~$c_\Sigma(V(F) \cap X) \leq b$.
	We show that~$\w_\Sigma(F) \leq d$.
	We distinguish two cases, whether~$F$ contains an incoming edge of~$v$ or not.
	Since~$(V(F),F)$ is a directed forest, $F$ contains at most one incoming edge of~$v$.
	We begin with the case that~$F$ does not contain an incoming edge of~$v$.
	We define~$F_i$ as the subset of edges in~$F$ in which all edges are directed toward a vertex in~$\Gamma_{\leq w_i}$.
	This is a partition of~$F$.
	Let~$b_i = c_\Sigma(V(F_i) \cap X)$.
	We partition~$Y$ in the following way.
	Let~$Y_i$ contain~$y \in Y$ if~$y \in V(F_i)$ and~$y \not\in V(F_j)$ for any~$j<i$.
	Clearly, these sets partition~$Y$.
	To see that~$F_i$ is~$Y_i$-compatible, we observe that~$F_i$ is~$Y \cap V(F)$-compatible, that~$Y_i \subseteq Y\cap V(F)$, and that if~$F_i$ is~$S$-compatible implies~$F_i$ being~$S'$-compatible for each~$S' \subseteq S$.
	Thus, $\w_\Sigma(F) = \sum_{i=1}^{\outdeg(v)} \w_\Sigma(F_i) \leq \sum_{i=1}^{\outdeg(v)} \DP_1[w_i,Y_i,b_i] \leq \DP_1'[v,Y,b] \leq \DP_1[v,Y,b] = d$.
	It remains to show the correctness when~$F$ contains~$p_v v$ for some parent~$p_v$ of~$v$.
	We define~$F_i$ as before and observe that~$p_v v$ is in none of the sets.
	But because~$F$ is~$Y$-compatible, there is an outgoing edge~$e$ of~$v$ in~$F$ and, without loss of generality, $e$ in~$F_1$.
	We define~$b_i$ and~$Y_i$ as above, but add~$v$ to~$Y_1$.
	Then, $Y_i$ as defined above is a partition of~$Y \setminus \{p_v\} \cup \{v\}$.
	We conclude $\w_\Sigma(F) = \sum_{i=1}^{\outdeg(v)} \w_\Sigma(F_i) \leq \sum_{i=1}^{\outdeg(v)} \DP_1[w_i,Y_i,b_i] \leq \w(p_v v) + \DP_1'[v,Y \setminus \{p_v\} \cup \{v\},b] \leq \DP_1[v,Y,b] = d$.
	Therefore, all table entries store the correct value.
\end{proof}

\subsection{Proof of \Cref{lem:MnTPD-table}}

\lemMnTPDtable*
\begin{proof}
	To see the correctness of the base case, we observe that the incoming edge of each leaf is in any switching tree and therefore also in the switching tree of minimum weight.
	This forces~$Z = \{v_x\}$ and does not allow the case with~$Z = \emptyset$.
	
	We show the correctness of the table~$\DP[v,Y,Z]$.
	That is, we show that if~$\DP[v,Y,Z] = d$, then there is a~$(Y,Z)$-compatible set of edges~$F \subseteq E(v)$ with~$\w_\Sigma(F) = d$.
	Afterward, we show that~$\DP[v,Y,Z] \le \w_\Sigma(F)$, for any~$(Y,Z)$-compatible set of edges~$F \subseteq E(v)$.
	As an induction hypothesis, we assume that~$\DP[w,\cdot,\cdot]$ stores the correct value for all $\Gamma$-children~$w$ of~$v$.

	Let~$\DP[v,Y,Z] = d$.
	Then, we have~$\DP'[v,Y,Z] = d$ or~$\min_{p_v \in Z}
	\w(p_v v) + \DP'[v,Y \setminus \{p_v\} \cup \{v\},Z \cup \{v\}] = d$, by \Recc{eq:bMnTPD-rec}.
	In the former case, there is a partition~$(Y_1,\dots,Y_{\outdeg(v)})$ of~$Y$ such that~$\DP[v,Y,Z] = \sum_{i=1}^{\outdeg(v)} \DP[w_i,Y_i,Z\cap\GW(w_i)]$ and there are~$(Y_i,Z \cap \GW(w_i))$-compatible sets of edges~$F_i \subseteq E(w_i)$ with~$\w_\Sigma(F_i) = \DP[w_i,Y_i,Z\cap\GW(w_i)]$.
	We define~$F = \biguplus_{i=1}^{\outdeg(v)} F_i$.
	Since the sets~$E(w_i)$ are pairwise disjoint, $F_i$ and~$F_j$ are disjoint whenever~$i \ne j$.
	It follows directly that~$F$ is~$(Y,Z)$-compatible and that~$\w_\Sigma(F) = \sum_{i=1}^{\outdeg(v)} \w_\Sigma(F_i) = d$.
	Now, assume that there is a vertex~$p_v \in Z\cap \parents(v)$ such that~$\DP'[v,Y \setminus \{p_v\} \cup \{v\},Z \cup \{v\}] = d - \w(p_v v)$.
	Thus, there is a partition~$(Y_1,\dots,Y_{\outdeg(v)})$ of~$Y \setminus \{p_v\} \cup \{v\}$ such that~$\DP[v,Y,Z] = \sum_{i=1}^{\outdeg(v)} \DP[w_i,Y_i,Z\cap\GW(w_i)]$ and there are~$(Y_i,Z \cap \GW(w_i))$-compatible sets of edges~$F_i \subseteq E(w_i)$ with~$\w_\Sigma(F_i) = \DP[w_i,Y_i,Z\cap\GW(w_i)]$.
	We define~$F = \{p_v v\} \cup \biguplus_{i=1}^{\outdeg(v)} F_i$.
	Without loss of generality, let~$v \in Y_1$ and thus~$F_1$ and~$F$ contain an outgoing edge of~$v$.
	For all other vertices in~$Y$, this is clear.
	Since~$E(w_i)$ does not contain incoming edges of~$v$, there is exactly one incoming edge, $p_v v$, of~$v$ in~$F$.
	Let~$z$ be a root of~$\Net[F_i]$ for any~$i \in [{\outdeg(v)}]$.
	If~$z$ is in~$Z$, then nothing remains to show.
	If~$z = v$, then, in~$\Net[F]$, the edge~$p_v v$ causes~$p_v$ to be a root, which is fine since~$p_v \in Z$.
	Thus, $\Net[F]$ is a tree with roots in~$Z$ and the set of leaves equals~$X(\Gamma_{\leq v})$.
	Because~$\w_\Sigma(F) = \w(p_v v) + \sum_{i=1}^{\outdeg(v)} \w_\Sigma(F_i) = \w(p_v v) + \DP'[v,Y \setminus \{p_v\} \cup \{v\},Z \cup \{v\}] = d$,
	this proves this direction of the statement.

	Now, let~$F \subseteq E(v)$ be a~$(Y,Z)$-compatible set of edges.
	We need to show that~$\DP[v,Y,Z] \leq \w_\Sigma(F)$.
	We distinguish two cases, depending on whether $F$ contains an incoming edge of~$v$ or not.
	If~$F$ does not contain an incoming edge of~$v$, then~$F_i$ with~$F_i := F \cap E(w_i)$ for each~$i \in [{\outdeg(v)}]$ form a partition of~$F$.
	We define~$Y_i$ as the set of vertices in~$Y \cap V(F_i)$ which do not occur in~$V(F_j)$ for any~$j<i$.
	Because the sets~$E(w_i)$ are pairwise disjoint, $\Net[F_i]$ is a tree with roots in~$Z\cap\GW(w_i)$ and whose leaf set is exactly~$X(\Gamma_{\leq w_i})$, for each~$i \in [{\outdeg(v)}]$.
	Therefore, $\DP[w_i,Y_i,Z\cap\GW(w_i)] \leq \w_\Sigma(F_i)$.
	We conclude with \Recc{eq:bMnTPD-rec}, that~$\DP[v,Y,Z] \leq \DP'[v,Y,Z] \leq \sum_{i=1}^{\outdeg(v)} \DP[w_i,Y_i,Z\cap\GW(w_i)] \leq \sum_{i=1}^{\outdeg(v)} \w_\Sigma(F_i) = \w_\Sigma(F)$.
	Now, let~$F$ contain~$p_v v$.
	Since~$\Net[F]$ is a tree, $F$~contains at most one incoming edge per vertex.
	Further, since~$p_v$ is not in~$\Gamma_{\leq v}$, we conclude that~$p_v$ is a root of~$\Net[F]$ and thus~$p_v \in Z$.
	We define~$F_i$ and~$Y_i$ as above and we remove~$p_v$ from all~$Y_i$.
	Because the leaves of~$\Net[F]$ are in~$X$ and~$v$ is an internal vertex, there is an outgoing edge~$e$ of~$v$ in~$F$.
	Without loss of generality, $e \in F_1$ and we add~$v$ to~$Y_1$.
	Then, the sets~$Y_i$ form a partition of~$Y \setminus \{p_v\} \cup \{v\}$.
	We use the arguments from above to see now that~$F_i$ is~$(Y_i,Z\cap\GW(w_i))$-compatible, for each~$i \in [{\outdeg(v)}]$.
	We conclude with \Recc{eq:bMnTPD-rec}, that~$\DP[v,Y,Z] \leq \w(p_v v) + \DP'[v,Y \setminus \{p_v\} \cup \{ v \},Z \cup \{v\}] \leq \w(p_v v) + \sum_{i=1}^{\outdeg(v)} \DP[w_i,Y_i,Z\cap\GW(w_i)] \leq \w(p_v v) + \sum_{i=1}^{\outdeg(v)} \w_\Sigma(F_i) = \w_\Sigma(F)$.
	Therefore, the table stores the correct values.
\end{proof}

\section{A reduction rule for removing some lowest reticulations}
\label{sec:rr-appendix}

We describe an additional reduction rule below, which reduces the number of reticulations in certain types of more complicated bi-connected components. See \Cref{fig:reticulation-removal} for an example.

\begin{rr}
	\label{rr:reticulation-removal}
	Apply \Cref{rr:deg2-chain,rr:cut-vertex} exhaustively.
	Let~$r$ be a reticulation without outgoing edges, let~$C$ be the parents of~$r$ that are chain vertices, and let~$P$ be the remaining parents of~$r$.
	Let~$P_C$ be the set of parents of~$C$.
	If~$P \cup P_C$ does not contain a reticulation and~$\anc(v^*) \supseteq P \cup P_C$, for some~$v^* \in P \cup P_C$---that is, each vertex~$v \in P \cup P_C$, $v\ne v^*$ has in $G$ a path to~$v^*$---then remove~$r$ and~$C$.
\end{rr}
\begin{figure}[t]
	\centering
    \includegraphics[width=0.90\linewidth]{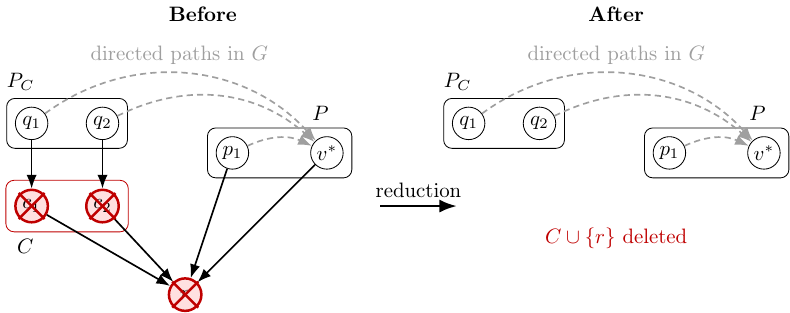}
	\caption{An illustrative example of \Cref{rr:reticulation-removal}.
		The dashed arrows denote directed paths in~$G$ that may leave the displayed subgraph.
		We note that these paths do not need to be disjoint.}
	\label{fig:reticulation-removal}
\end{figure}
\begin{lemma}
	\label{lem:reticulation-removal}
	\Cref{rr:reticulation-removal} is correct and can be applied exhaustively in~$\Oh(n^2)$~time.
	
	Let~$G'$ be the resulting DAG and let~$\Gamma'$ be an optimal tree-extension.
	Let~$Q$ be a path with vertices~$C$ in an arbitrary order and~$r$ at last with root~$q$.
	Then, connecting~$\Gamma'$ and~$Q$ with an edge~$v^* q$ gives an optimal tree-extension~$\Gamma$ of node scanwidth~$\max(\nswwithoutN (\Gamma'),\indeg(r))$.
\end{lemma}
\begin{proof}
	$\Gamma$ is clearly a tree.
	To see that~$\Gamma$ is a tree-extension, it remains to consider edges~$uw$ with~$u\in V(\Gamma')$ and~$w\in V(Q)$, since the others are clear because~$\Gamma'$ and~$Q$ are tree-extensions.
	Since~$u \in P\cup P_C$, there is a (possibly empty) path from~$u$ to~$v^*$.
	It directly follows that~$u$ has a path to~$w$ in~$\Gamma$, since~$v^* q \in E(\Gamma)$ and~$Q$ is a path.
	By \Cref{obs:deg2-tree}, we know that there is an optimal tree-extension in which~$P_C$ are directly above~$r$.
	
	It remains to show that the bound is correct.
	Each vertex~$v\in V(\Gamma')$ satisfies~$\GW(v,\Gamma) = \GW(v,\Gamma')$.
	Further, $\GW(r,\Gamma) = P \cup C$ and has size~$\indeg(r)$ and the bag of each vertex in~$P_C$ removes and adds exactly one vertex in comparison to its~$\Gamma$-child.
	It follows that the bound on the node scanwidth is correct.
	We conclude that~$\Gamma$ is optimal with \Cref{obs:in-deg,obs:subnet}.
	
	Concerning the running time, we observe that \Cref{rr:reticulation-removal} can be applied at most once per reticulation and for each reticulation~$r$ checking the condition takes~$\Oh(\indeg(r)) \in \Oh(n)$ time.
	We further observe that if a reticulation does not satisfy the condition at some point, then it will also not satisfy the condition later, and so we can iterate once over the vertices.
\end{proof}

\section{An ILP for computing node scanwidth}
\label{sec:ILP}

We use five kinds of variables to find an optimal tree-extension~$\Gamma$.
The variables are set up to capture the following information.
Variable~$s$ takes the node scanwidth of~$\Gamma$ and takes integer values.
All following variables will take binary values, 0~or~1.
For each pair of vertices~$u,v \in V$, we define the following three variables.
Variable~$x_{u,v}$ is set to~$1$ if and only if~$uv$ is an edge in~$\Gamma$.
Variable~$y_{u,v}$ is set to~$1$ if and only if~$v$ can be reached from~$u$ in~$\Gamma$.
Variable~$z^u_v$ is set to~$1$ if and only if~$u \in \GW(v)$ in~$\Gamma$.
To compute transitivity, we resort to a standard trick and define an auxiliary variable~$\alpha_{u,v,w}$, for each triple~$u,v,w$ of vertices, which is set to~1 if and only if there is a path from~$u$ to~$v$ in~$\Gamma$ and an edge~$vw$ in~$\Gamma$.

The \ILP is defined as follows. We write $\delta_{uw \in E} = 1$ if~$uw \in E$ and $\delta_{uw \in E} = 0$ otherwise.
\begin{align}
	& \text{minimize} & s\\
	\label{eq:GW}
	& \text{subject to} & s &~\ge \sum_{u \in V} z^u_v & \forall v \in V\\
	\label{eq:necessary-reachability}
	&& 1 &~\le y_{u,v} & \forall uv \in E\\
	\label{eq:bag-def}
    && 1 &~= y_{v,v} & \forall v \in V\\
	\label{eq:y-basecase}
	&& z^u_v &~\ge y_{u,v} + y_{v,w} + \delta_{uw \in E} - 2 & \forall u,v,w \in V, u \neq v\\
	\label{eq:tree-edges}
	&& |V|-1 &~= \sum_{u\in V}\sum_{v\in V} x_{u,v} & \\
	\label{eq:tree-incoming}
	&& 1 &~\ge \sum_{u\in V} x_{u,v} & \forall v \in V\\
    \label{eq:x_basecase}
	&& 0 &~= x_{v,v} & \forall v \in V\\
	\label{eq:cycle-free}
	&& 1 &~\ge y_{u,v} + y_{v,u} & \forall u,v \in V, u\ne v\\
	\label{eq:alpha-def}
	&& \alpha_{u,v,w} &~\ge y_{u,v} + x_{v,w} -1 & \forall u,v,w \in V\\
	\label{eq:alpha1}
	&& y_{u,w} &~\ge \alpha_{u,v,w} & \forall u,v,w \in V, u\ne w\\
	\label{eq:alpha2}
	&& y_{u,w} &~\le \sum_{v\in V} \alpha_{u,v,w} & \forall u,w \in V, u \neq w\\
	\label{eq:alpha3}
	&& \alpha_{u,v,w} &~\le y_{u,v} & \forall u,v,w \in V\\
	\label{eq:alpha4}
	&& \alpha_{u,v,w} &~\le x_{v,w} & \forall u,v,w \in V\\
	&& x_{u,v}; y_{u,v}; \alpha_{u,v,w}; z^u_v &~\in \{0,1\} & \forall u,v,w \in V\\
	&& s &~\in \mathbb{N}_0
\end{align}

We continue with explaining the correctness of the \ILP formulation.
By definition, the node scanwidth is the size of the maximum bag.
Especially, it is larger than any bag size, which we capture in~\Cref{eq:GW}.
For any edge~$uv$ of the input graph, in any tree-extension, $v$ must be reachable from~$u$, which is forced by~\Cref{eq:necessary-reachability}.
Now, we know that every DAG with~$n$ vertices in which all vertices have in-degree at most~1 is a directed forest.
If additionally the DAG has~$n-1$ edges, then it has to be connected and thus a tree.
These conditions are enforced with~\Cref{eq:tree-edges,eq:tree-incoming,eq:x_basecase}.
To ensure acyclicity, we have~\Cref{eq:cycle-free}.
\Cref{eq:alpha-def} follows our definition of~$\alpha$.
The definition of $\alpha$ is necessary, because we want to define~$y_{u,w} = \max_{v \in V} y_{u,v} \cdot x_{v,w}$.
This non-linearity cannot be directly represented in an \ILP.
So, we need to define~$\alpha_{u,v,w} = y_{u,v} \cdot x_{v,w}$ and resort to the standard technique for such a case, as displayed in~\Cref{eq:alpha1,eq:alpha2,eq:alpha3,eq:alpha4}.

This completes the definition of the \ILP.

\newpage
\thispagestyle{empty}
\renewcommand\thepage{\arabic{page}}

\end{document}